\documentclass[twoside,leqno,twocolumn,9pt]{extarticle}

\usepackage{tabu}                     
\usepackage{booktabs}                 
\usepackage{lipsum}                   
\usepackage{mwe}                      
\usepackage{graphicx}                 
\usepackage{makecell}                 
\usepackage{multirow}                 
\usepackage{threeparttable}           
\usepackage{subcaption}               
\usepackage{url}
\usepackage[table,xcdraw]{xcolor}     
\usepackage{caption}
\usepackage{xcolor}
\usepackage{fancyhdr}
\fancyhfoffset{0pt}
\usepackage{amsmath,amssymb,bm}
\usepackage{amsthm}
\usepackage{siunitx}
\usepackage[english]{babel}

\usepackage[ruled,vlined]{algorithm2e}
\LinesNumbered
\SetKw{KwStep}{\textbf{Step}}         
\SetKw{KwInit}{\textbf{Init}}         

\newtheorem{theorem}{Theorem}

\newtheorem{lemma}{Lemma}
\newtheorem{proposition}{Proposition}

\newtheorem{assumption}{Assumption}

\DeclareMathOperator*{\argmin}{arg\,min}

\usepackage{etoolbox} 
\BeforeBeginEnvironment{algorithm}{\vspace{-0.5em}} 
\AfterEndEnvironment{algorithm}{\vspace{-0.5em}}   

\SetAlCapSkip{0.5em} 
\usepackage{geometry}
\usepackage{enumitem}
\title{From GPUs to RRAMs: Distributed In-Memory Primal–Dual Hybrid Gradient Method for Solving Large-Scale Linear Optimization Problems}

\fancypagestyle{plain}{%
  \fancyhf{} 
  \fancyfoot[C]{\thepage}
}

\author{
Huynh Q. N. Vo\thanks{School of Industrial Engineering and Management, Oklahoma State University, Stillwater, OK, USA} \thanks{Energy Systems and Infrastructure Assessment Division, Argonne National Laboratory, Lemont, IL, USA}
\and
M. T. R. Chowdhury\thanks{Department of Electrical and Computer Engineering, Wayne State University, Detroit, MI, USA} \footnotemark[2]
\and
Paritosh Ramanan\footnotemark[1]
\and
Gözde Tutuncuoglu\footnotemark[3]
\and
Junchi Yang\footnotemark[2]
\and
Feng Qiu\footnotemark[2]
\and
Murat Yildirim\thanks{Department of Industrial and Systems Engineering, Wayne State University, Detroit, MI, USA}
}

\date{}
\begin{document}
\maketitle

\begin{abstract}
The exponential growth of computational workloads is surpassing the capabilities of conventional architectures, which are constrained by fundamental limits. In-memory computing (IMC) with RRAM provides a promising alternative by providing analog computations with significant gains in latency and energy use. However, existing algorithms developed for conventional architectures do not translate to IMC, particularly for constrained optimization problems where frequent matrix reprogramming remains cost prohibitive for IMC applications. Here we present a distributed in-memory primal–dual hybrid gradient (PDHG) method, specifically co-designed for arrays of RRAM devices. Our approach minimizes costly write cycles, incorporates robustness against device non-idealities, and leverages a symmetric block-matrix formulation to unify operations across distributed crossbars. We integrate a physics-based simulation framework called MELISO+ to evaluate performance under realistic device conditions. Benchmarking against GPU-accelerated solvers on large-scale linear programs demonstrates that our RRAM-based solver achieves comparable accuracy with up to three orders-of-magnitude reductions in energy consumption and latency. These results demonstrate the first PDHG-based LP solver implemented on RRAMs, showcasing the transformative potential of algorithm–hardware co-design for solving large-scale optimization through distributed in-memory computing.
\end{abstract}

\smallskip\noindent\textbf{Keywords:} In-Memory Computing, Distributed Constrained Optimization, Primal-Dual Hybrid Gradient Method

\section{Introduction}

The exponential growth in computational workloads is imposing significant challenges for sustainable computing \cite{10838774}, while conventional computing architectures are struggling to keep pace due to diminishing returns from Moore’s law \cite{doi:10.1126/science.ade2191, 10454417}. A central obstacle for conventional computing architectures is the von Neumann bottleneck—the physical separation of memory and processing—which imposes substantial energy and latency costs through continual data movement during sequential computational tasks \cite{zou_breaking_2021,6634083}. Emerging in-memory computing (IMC) architectures, such as resistive random-access memory (RRAM) device arrays, resolve this obstacle in two key aspects: (i) by co-locating memory and computation, RRAM eliminates costly data transfers, and (ii) RRAMs enable a single analog step for performing matrix vector multiplication (MVM) operations when configured in a crossbar array architecture. Collectively, these aspects offer not only significant gains in latency and energy use for a range of computing tasks but also introduce a new set of challenges stemming from device non-idealities and intrinsic differences in RRAM operations. Hence, computing methods developed for conventional computing do not translate well to IMC applications, which require the development of new algorithms to ensure both computational accuracy and scalability for fundamental computing tasks. This is particularly true for constrained optimization problems, where frequent write operations—trivial in conventional computing systems—become cost prohibitive for IMC applications. 

Among the various IMC technologies, RRAM devices stand out as an ideal candidate for computing applications like constrained optimization due to their scalability potential, back-end-of-the-line (BEOL) compatibility with existing CMOS technology, and nonvolatile characteristics \cite{doi:10.1021/acs.chemrev.4c00845}. RRAM devices operate through reversible resistance switching \cite{chowdhury2025} in metal-oxide-metal structures \cite{10187322,10477696}, where the conductance state can be precisely controlled by analog input and retained without power, enabling both data storage and analog computation within the same physical device. By satisfying key development criteria—including high density, rapid read/write capabilities, fast access speeds, and low power consumption—RRAM positions itself as a promising candidate to potentially replace traditional storage devices and achieve significant computing power enhancements \cite{10454367}.

RRAM device arrays inherently enable efficient MVMs by leveraging Ohm’s and Kirchhoff’s law \cite{https://doi.org/10.1002/adfm.202310193,doi:10.1021/acsaelm.4c00199}. More specifically, RRAMs perform MVM by applying an input voltage to the rows of the crossbar that triggers an analog multiplication with the conductance states that are located at the intersection points of the array, which are then summed over the arrays. In the context of an MVM operation, the conductance states of the array are analogous to matrix entries, the vector entries are given by the input voltage, and the output current gives the resulting vector. This way of conducting MVM computation is in sheer contrast to conventional architectures that rely on sequential multiply–accumulate operations between memory and processing units. These factors collectively provide significant improvements for RRAMs in scientific computing compared to their digital equivalents, reported to reduce energy use 160 times and 128 times \cite{doi:10.1126/sciadv.adr6391}. Since MVMs are the fundamental operations for the majority of linear operations and optimization algorithms \cite{10.1145/3615679, zhao_block-wise_2024}, large-scale optimization problems are often accelerated on GPUs due to their massive parallelism, enabling fast matrix–vector operations. However, the state-of-the-art GPUs still suffer from the von Neumann bottleneck. On the contrary, RRAM can enable energy-efficient low-latency computing architectures and offers a transformative path toward sustainable computing for optimization tasks \cite{le_gallo_mixed-precision_2018}. Capitalizing on this promise, this paper develops a novel framework specifically designed for RRAM devices to solve distributed constrained optimization problems.

\subsection{Related Work}
In this section, we revise two streams of research that, until now, have largely been developed with limited interaction. The first stream focuses on emerging optimization methods in conventional computing that motivate some of our method development. The second stream focuses on computing in RRAM-based devices that showcase opportunities and challenges unique to in-memory architectures. Together, these research streams form the foundation for our proposed method.

\noindent\textit{First-order methods for linear programs:} Large-scale constrained linear programs (LP) push the limits of classical solvers based on simplex and interior-point methods due to heavy sparse matrix factorizations that are performed at each iteration~\cite{lu2024cupdlpjlgpuimplementationrestarted}. These factorizations become a bottleneck: they can consume significant memory—since sparse constraint matrices often yield dense LU/Cholesky factors—and are difficult to parallelize.
In fact, for very large instances, memory blow-up from factorization often causes out-of-memory errors even when the raw matrix data fits in memory~\cite{lu2024cupdlpjlgpuimplementationrestarted}.
First-order methods such as primal–dual hybrid gradient (PDHG), initially designed for image processing applications~\cite{chambolle2011first, condat2013primal, he2012convergence, zhu2008efficient}, avoid this pitfall by using only matrix–vector multiplications instead of matrix inversions.
Each PDHG iteration has a low per-iteration cost and is matrix-free, meaning the algorithm only needs to access the constraint matrix via multiply operations.
This leads to multiple advantages in large-scale settings: the solver's memory footprint is essentially just the problem data, and the simple, regular operations can exploit modern parallel hardware much more effectively than the sequential steps of simplex or interior-point solvers. These advantages also enable an efficient use of GPU for solving LP models, which is not practical for more conventional solution methods such as simplex and interior point methods.

Recent research has demonstrated that with proper enhancements, first-order methods can attain high solution accuracy on LPs and even outperform classical solvers on certain problems.
For example, the works by Lu et al. on PDLP~\cite{lu2024cupdlpjlgpuimplementationrestarted} applied PDHG to LPs with additional techniques like diagonal preconditioning, adaptive step sizing, and restarting, enabling it to solve many benchmark LPs faster than the ADMM-based splitting conic solver (SCS)~\cite{o2021operator} and to beat commercial solvers—e.g., Gurobi~\cite{gurobi}—on a very large PageRank LP instance~\cite{NEURIPS2021_a8fbbd3b}.
Interestingly, developers of commercial LP solvers have also begun integrating PDHG: Gurobi reported that an experimental PDHG algorithm was potentially competitive with its highly optimized simplex and interior-point methods on certain large-scale LPs.~\cite{touretzky2025gurobi}.
These studies underscore that a matrix–vector-centric iterative solver can scale to problem sizes beyond the reach of factorization-based solvers, especially when paired with parallel architectures.

\noindent\textit{RRAM-Based Computing for Optimization:} 
The underlying physical process of RRAM-based computing renders a $n \times n$ matrix–vector multiplication effectively constant in latency—i.e., $\mathcal{O}(1)$ in problem size—in an ideal RRAM implementation~\cite{gokmen2016acceleration}. The energy and latency benefits provide significant benefits for a range of optimization tasks. 
These efforts examined the feasibility of analog LP solvers using RRAMs, leveraging their intrinsic MVM capabilities to accelerate linear solvers \cite{9540988, 9597387, 10092813,li2025voltagecontrolledoscillatormemristorbasedanalog, 11099332}. Cai et al. introduced a low-complexity LP solver that exploits RRAM dynamics for real-time optimization, achieving significant reductions in computation latency and power consumption \cite{7905500, CAI201862}. Shang et al. extended this paradigm by proposing an analog recursive computation circuit tailored for LP, integrating RRAMs with operational amplifiers to realize iterative updates in hardware \cite{9094179}. 
Di Marco et al. further explored neural network-inspired formulations, embedding LP constraints into RRAM-based Hopfield networks to solve both linear and quadratic programs with high fidelity \cite{9121721}. Recent work by Hizzani et al. introduced higher-order Hopfield solvers, outperforming traditional digital methods in speed and energy metrics \cite{10558658}. These designs capitalize on the nonlinearity and memory retention of RRAMs to encode complex optimization landscapes. Liu et al. introduced a generalized optimization framework based on the Alternating Direction Method of Multipliers (ADMM), designed to tackle LPs and quadratic programs (QPs) by offloading linear equation solves to RRAM crossbar hardware \cite{7858420}. Building on this foundation, they extended the approach for AI-driven workloads by embedding LP solvers into a larger RRAM-based computing architecture \cite{8288635}.

Existing RRAM-based optimization methods have certain limitations. First, existing methods have largely been confined to small-scale problems. Second, the main focus for these works is on circuit- or algorithm-level effects while neglecting device-level variability and non-idealities, which constitute the central challenge for reliable IMC implementations. Third, these methods tend to be architecturally rigid and lack robustness, preventing their use as generalized linear solver platforms. Crucially, a comparative analysis against benchmarks such as GPU implementations and commercial solvers such as Gurobi is also missing in prior state-of-the-art.

Regardless, these studies underscore a shift toward in-memory analog computing for optimization, where RRAMs serve not merely as passive elements but as active computational primitives. However, using RRAM hardware for precise computation tasks like constrained optimization is a challenge.
Analog results are subject to device non-idealities that introduce cumulative errors to the final results~\cite{Gulafshan2025,10766540}.
Furthermore, RRAM crossbars can only represent non-negative weights (conductances cannot be negative), so encoding an arbitrary real matrix requires additional tricks such as using dual crossbars or adding offset columns/rows to handle negative values.
Researchers have developed circuit-level and algorithmic techniques to mitigate these issues for related computational tasks ~\cite{vo2025harnessing,10281389,10818995,doi:10.1126/science.adi9405,9779296,10536974}. This convergence of circuit design, algorithmic innovation, and device physics \cite{chowdhury2025} has positioned RRAM LP solvers as a compelling alternative to von Neumann architectures, but they require careful management of precision, device non-idealities, and computational scalability issues to realize those benefits in practice.

\subsection{Contributions}
We propose a distributed in-memory PDHG method for solving large-scale optimization problems. To our knowledge, this is the first demonstration of a PDHG-based LP solver that incorporates device-level modeling through MELISO+ and NeuroSim+~\cite{8268337} to capture variability effects. In contrast to existing approaches, our method scales to substantially larger problem instances and provides a versatile solver design applicable for a range of constrained optimization models. Our work highlights the impact of a co-design approach, where algorithmic development and hardware constraints have been addressed jointly. The key contributions are as follows:

\begin{itemize}[leftmargin=*]
\item We formulate a distributed in-memory PDHG method designed from the ground up for RRAM devices. 
To this end, the proposed PDHG method offers the following key advantages:
\begin{itemize}[leftmargin=*]
\item i) Our method significantly reduces the number of required updates to the matrices in order to reduce the write cycles, which are particularly costly in RRAMs. In contrast, the model uses multiple iterations with a fixed matrix and changing vectors, which is much faster for RRAMs compared to conventional computing paradigms.
\item ii) An enhanced PDHG formulation is developed to work reliably in noisy computation tasks in an effort to mitigate the inaccuracies resulting from RRAM device non-idealities. 
\item iii) A distributed computing framework is developed for RRAMs that decompose and solve the problem by considering device characteristics and sizes of an RRAM array.
\item iv) An implementation scheme is developed to encode and solve any class of LPs using RRAM devices. 
\end{itemize}
\item We provide theoretical convergence bounds for PDHG under analog noise in RRAM-based implementations, which guarantee that in-memory PDHG remains convergent on noisy RRAM hardware under mild conditions, and the rate is provably no worse than conventional computing implementation up to a fixed error term. Other theoretical results also support the significant changes made to the PDHG algorithm to improve efficiency by exploiting the problem structure. 
\item We develop a simulation and evaluation platform for RRAM devices to test the performance of the algorithms using extensive physics-based simulation models (i.e., MELISO+ and NeuroSim+~\cite{8268337}), extended onto a distributed computing environment. The proposed platform provides the capability to faithfully evaluate an extensive set of algorithm performance metrics related to latency and solution accuracy.
\end{itemize}

The remainder of the paper is organized as follows. Section 2 reviews the preliminaries of linear programming and the PDHG method. Section 3 develops our in-memory PDHG formulation for RRAM devices. Section 4 introduces the proposed hardware–algorithm co-design framework. Section 5 presents experimental results on accuracy, energy, and latency. Section 6 discusses scalability and robustness, and Section 7 concludes with key findings and future directions. 

\section{Preliminaries}
The section introduces the LP formulation and the corresponding saddle-point problem, which are used within a PDHG solution algorithm. We conclude the section with common methods for algorithmic enhancements of PDHG in practice.

\subsection{Saddle-point formulation}
In our framework, we propose a methodology to solve a general class of LPs that can be denoted as:
\begin{equation}
    \begin{aligned}
    \min_{\mathbf{x}\in\mathbb{R}^n}\quad & \mathbf{c}^\top \mathbf{x} \\
        \text{s.t.}\quad & \mathbf{G}\mathbf{x} \geq \mathbf{h},\quad \mathbf{A}\mathbf{x} = \mathbf{b},\quad \ell \le x_i \le u, \forall x_i \in \mathbf{x}
    \end{aligned}
\label{eq:boxed-lp}
\end{equation}

where $\mathbf{G}\in\mathbb{R}^{m_1\times n}$, $\mathbf{K}\in\mathbb{R}^{m_2\times n}$, $\mathbf{c}\in\mathbb{R}^n$, $\mathbf{h}\in\mathbb{R}^{m_1}$, $\mathbf{b}\in\mathbb{R}^{m_2}$, and $\ell,u\in \mathbb{R}$.

Projecting this model directly onto the feasible polytope in \eqref{eq:boxed-lp} is generally expensive for first‑order methods (FOMs), so we dualize the linear constraints to obtain the following saddle point problem:

\begin{equation}
    \min_{\mathbf{x}\in \mathcal{X}}\;
\max_{\mathbf{y}\in \mathcal{Y}}\; L(\mathbf{x},\mathbf{y})
    := \mathbf{c}^\top \mathbf{x} - \mathbf{y}^\top \mathbf{K} \mathbf{x} + \mathbf{q}^\top \mathbf{y}
    \label{eq:saddle}
\end{equation}
with
\begin{equation*}
    \begin{cases}
        \mathbf{K}^\top= \begin{bmatrix} \mathbf{G}^\top & \mathbf{A}^\top\end{bmatrix}, \quad \mathbf{q}^\top=\begin{bmatrix} \mathbf{h}^\top & \mathbf{b}^\top\end{bmatrix} \\
        \mathcal{X}:=\{\mathbf{x}\in\mathbb{R}^n | \ell\le x_i\le u, \forall x_i \in \mathbf{x}\}, \\
        \mathcal{Y}:=\{\mathbf{y}\in\mathbb{R}^{m_1+m_2}|y_{1:m_1}\ge 0\}.
\end{cases}
\end{equation*}

Consequently, a saddle point of \eqref{eq:saddle} recovers an optimal primal–dual solution to \eqref{eq:boxed-lp}. Going forward, We will use LPs in their standard form upon suitable projection:
\begin{equation}
    \begin{aligned}
        \min_{\mathbf{x}\in\mathbb{R}^n}\quad & \mathbf{c}^\top \mathbf{x} \\
        \text{s.t.}\quad & \mathbf{K}\mathbf{x} = \mathbf{b},\quad x_i \geq 0, \forall x_i \in \mathbf{x}
    \end{aligned}
\label{eq:standard-lp}
\end{equation}

Thus, the projections in \eqref{eq:standard-lp} become element-wise non-negative for $\mathbf{x}$, and unbounded for $\mathbf{y}$.

\subsection{Derivation of PDHG for Linear Problems}\label{subsection:PDHG_derivation}
Given the saddle point problem, \eqref{eq:saddle}, described in the $f$–$g$ saddle template:

\begin{equation*}
    \begin{aligned}
        &\min_{\mathbf{x}} \max_{\mathbf{y}}\ \bigl[f(\mathbf{x}) - \langle \mathbf{K}\mathbf{x},\mathbf{y}\rangle - g(\mathbf{y})\bigr], \\
        &\text{where} \quad
        \begin{cases}
            f(x):= \mathbf{c}^\top \mathbf{x} + \delta_\mathcal{X}(\mathbf{x}) \\
            g(y):=\delta_\mathcal{Y}(\mathbf{y}) - \mathbf{q}^\top \mathbf{y} \\
            \delta_{\mathcal{S}}:= \text{the indicator of set } \mathcal{S}.
        \end{cases}
    \end{aligned}
\end{equation*}

The primal–dual hybrid gradient (PDHG), a.k.a. the Chambolle–Pock method~\cite{chambolle2011first}, is the following proximal splitting with extrapolation parameter $\theta\in[0,1]$, commonly selected to be $\theta=1$:
\begin{align}
  \mathbf{x}^{\{k+1\}}
  &= \argmin_{\mathbf{x}}
     \left\{
       f(\mathbf{x})
       + \tfrac{1}{2\tau}
         \left\|
\mathbf{x} - \left(\mathbf{x}^{\{k\}} - \tau \mathbf{K}^\top \mathbf{y}^{\{k\}}\right)
         \right\|_2^2
     \right\}, \label{eq:pdhg-x-prox} \\
  \bar{\mathbf{x}}^{\{k+1\}}
  &= \mathbf{x}^{\{k+1\}}
     + \theta \big|_{\theta = 1} \left(\mathbf{x}^{\{k+1\}} - \mathbf{x}^{\{k\}}\right)
     \label{eq:pdhg-extrap}, \\
  \mathbf{y}^{\{k+1\}}
  &= \argmin_{\mathbf{y}}
     \left\{
       g(\mathbf{y})
       + \tfrac{1}{2\sigma}
         \left\|
\mathbf{y} - \left(\mathbf{y}^{\{k\}} + \sigma\, \mathbf{K}\,\bar{\mathbf{x}}^{\{k+1\}}\right)
         \right\|_2^2
     \right\}.
 \label{eq:pdhg-y-prox}
\end{align}
with primal and dual step sizes—that is, $\tau,\sigma>0$—coupled by $\tau\sigma\|\mathbf{K}\|_2^2<1$.

Consider two proximal operations in PDHG:

\begin{itemize}
    \item Since $f(\mathbf{x})=\mathbf{c}^\top \mathbf{x} + \delta_\mathcal{X}(\mathbf{x})$, \eqref{eq:pdhg-x-prox} reduces to a shifted projection:
        \begin{equation*}
            \mathbf{x}^{\{k+1\}}= \operatorname{proj}_\mathcal{X}\!\bigl\{\mathbf{x}^{\{k\}} - \tau(\mathbf{c} - \mathbf{K}^\top \mathbf{y}^{\{k\}})\bigr\}.
        \end{equation*}
    \item For the dual step, $g(\mathbf{y}) = \delta_\mathcal{Y}(\mathbf{y})-\mathbf{q}^\top \mathbf{y}$, so \eqref{eq:pdhg-y-prox} is a quadratic objective over $\mathcal{Y}$ with a linear term. Completing the square yields:
        \begin{equation*}
            \mathbf{y}^{\{k+1\}} = \operatorname{proj}_\mathcal{Y}\!\bigl\{\mathbf{y}^{\{k\}} + \sigma\,(\mathbf{q} - \mathbf{K}\mathbf{\bar{x}}^{\{k+1\}})\bigr\}.
        \end{equation*}
\end{itemize}

Given that we selected $\theta=1$ such that $\mathbf{\bar{x}}^{\{k+1\}}=2\mathbf{x}^{\{k+1\}}-\mathbf{x}^{\{k\}}$, the iterative updates in the PDHG method for \eqref{eq:saddle} become:
\begin{equation}
\boxed{
    \begin{aligned}
        \mathbf{x}^{\{k+1\}} &\leftarrow \operatorname{proj}_{\mathcal{X}}\!\bigl\{\mathbf{x}^{\{k\}} - \tau\,(\mathbf{c} - \mathbf{K}^\top \mathbf{y}^{\{k\}})\bigr\},\\[3.0pt]
        \mathbf{y}^{\{k+1\}} &\leftarrow \operatorname{proj}_{\mathcal{Y}}\!\bigl\{\mathbf{y}^{\{k\}} + \sigma\Big[\mathbf{q} - \mathbf{K}(2\mathbf{x}^{\{k+1\}}-\mathbf{x}^{\{k\}})\Big]\bigr\}.
    \end{aligned}
}
\label{eq:pdhg-lp}
\end{equation}

Proposed in~\cite{lu2024cupdlpjlgpuimplementationrestarted,NEURIPS2021_a8fbbd3b}, the primal and dual step sizes can be re-parameterized as:
\begin{equation*}
    \tau=\frac{\eta}{\omega},\qquad \sigma=\eta\,\omega,\qquad \eta>0,\ \omega>0,
\end{equation*}

where the \emph{step‑size} $\eta$ sets the overall scale and the \emph{primal weight} $\omega$ balances primal and dual progress. Consequently, the coupling $\tau\sigma\|\mathbf{K}\|_2^2<1$ becomes $\eta<\|\mathbf{K}\|_2^{-1}$. Most importantly, the cost of \eqref{eq:pdhg-lp} is dominated by the two matrix–vector multiplies $\mathbf{K}^\top \mathbf{y}^{\{k\}}$ and $\mathbf{K}(2\mathbf{x}^{\{k+1\}} - \mathbf{x}^{\{k\}})$, meaning no linear solves are required.
To satisfy the step-size condition, the PDHG pipeline typically involves two steps, where the first focuses on estimating the spectral norm $\mathbf{K}$, and the second step uses this estimation to execute the PDHG algorithm.

\subsection{Algorithmic Enhancements for PDHG}
To improve the performance of PDHG for large-scale LPs, several practical components are introduced in practice:

\begin{itemize}
    \item \textit{Preconditioning}: The efficacy of FOMs depends on the conditioning of the underlying problem, represented by the conditioning number of $\mathbf{K}$. By employing the preconditioning step—via the Ruiz rescaling~\cite{ruiz2001scaling} and/or the method proposed in~\cite{pock2011diagonal}, we improve the geometry of $\mathbf{K}$ and the projections.
    
    \item \textit{Adaptive step sizes and primal‑weight updates}: In practical applications, the theoretical step-size--that is, $1/\|\mathbf{K}\|_2$—is deemed too conservative. Thus, these approaches, proposed in ~\cite{NEURIPS2021_a8fbbd3b,lu2024cupdlpjlgpuimplementationrestarted}, are employed to keep $\tau\sigma\|\mathbf{K}\|_2^2<1$ while balancing primal/dual progress.
    
    \item \textit{Adaptive restart}: Proposed in~\cite{lu2024cupdlpjlgpuimplementationrestarted,applegate2023faster}, the adaptive restart is employed to enhance convergence by regaining fast (often near‑linear) progress on sharp instances.
    
    \item \textit{Certificates and infeasibility detection:} By monitoring the sequences $\mathbf{d}^{\{k\}} = \mathbf{z}^{\{k+1\}} - \mathbf{z}^{\{k\}}$—with $\mathbf{z}^{\{k\}} = \begin{bmatrix} \mathbf{x}^{\{k\}} &\mathbf{y}^{\{k\}} \end{bmatrix}$, or the normalized average $2\mathbf{\bar{z}}_k/(k+1)$—with $\mathbf{\bar{z}}_k = (\mathbf{z}^{\{k\}}-\mathbf{z}^{\{0\}})/2$, the PDHG provides an infeasibility certificate for the problem of interest~\cite{applegate2024infeasibility}.
\end{itemize}

When $\|\mathbf{K}\|_2$ is unknown, a simple two-sided power iteration (PI) using the same $\mathbf{K}$ and $\mathbf{K}^\top$ is employed to estimate $\|\mathbf{K}\|_2$ robustly.
Once the estimation of $\|\mathbf{K}\|_2$ is obtained, the theoretical step sizes—namely, $(\tau,\sigma)$—are derived by setting $\eta \approx 0.95/\|K\|_2$ (for safeguarding) and then from $(\eta,\omega)$.

\section{In-Memory Framework for Distributed Large-Scale Linear Optimization}
In this section, we outline our proposed in-memory constrained optimization framework. Fig.~\ref{fig:overview} provides an overview of the framework.
Our framework has two primary functions. First, it compares the solution of a benchmarking linear problem obtained from a user-specified RRAM-based solver with the corresponding ground truth obtained from commercial solvers such as Gurobi.
The second function is to provide a detailed methodology and corresponding simulation for in-memory implementation of the PDHG method, including results for comparing with other accelerators such as GPUs.
To this end, the process starts with model preparation that includes preprocessing and preconditioning of the linear problem.
Following this, we propose an in-memory constrained optimization implementation that builds on a two-step approach.
The first step implements an iterative algorithm to estimate certain problem's parameters to ensure convergence.
The second step employs the resulting parameters to implement the PDHG method on RRAM devices.
In both methods, we introduce several specific adaptations required for implementing them in RRAM systems. Both the first and second steps are built on a full-stack module called MELISO+, a benchmarking framework designed to evaluate large-scale linear operations in RRAMs~\cite{10766540,vo2025harnessing}.
Using MELISO+ as a subprocess, we demonstrate how data can be encoded into physical devices and how algorithmic modifications mitigate device non-idealities that would otherwise cause solution divergence.

Before launching the RRAM-based solver, we perform the model preparation steps in the CPU, which involves deploying the problem and converting the model to its canonical form.
Then, we perform Ruiz rescaling and diagonal preconditioning, as proposed in~\cite{lu2024cupdlpjlgpuimplementationrestarted}, to improve the performance of PDHG.
These processes are commonly used in PDHG implementations in the literature and will be used as a first step in both our proposed method and the method implemented in GPU, which serves as our baseline for comparison.

In-memory implementation requires significant changes, novel algorithms, and theoretical guarantees to enable efficient and reliable implementation in RRAM devices. The method development revolves around three significant challenges and opportunities associated with RRAM devices: First, embedding data into RRAM devices is both energy- and time-intensive, making it imperative to design methods that minimize the number of write operations into the matrices. Unlike the matrices, changing the vectors is trivial and only requires altering the supplied current. Thus, it is important to develop methods that minimally change the matrix entries and freely alter vector elements. Second, device non-idealities—such as cycle-to-cycle and device-to-device variability—create a significant discrepancy between the programmed and the realized conductance states, requiring algorithmic robustness against inherent write errors. Third, once data are embedded, RRAM devices enable highly efficient matrix–vector multiplications, delivering orders-of-magnitude improvements in latency and energy consumption compared to traditional architectures.
This makes the cost of matrix vector multiplication significantly lower.

\subsection{Encoding constraint matrices}
The low latency and energy MVM operation in RRAMs allows iterative optimization methods that rely on repeated MVMs to be executed at scale, provided that only small and incremental modifications to the stored matrices are required.

\begin{algorithm}[!htbp]
    \caption{\textsc{buildSymBlock}$(\mathbf{K})$}
    \label{alg:build-sym-block}
    \DontPrintSemicolon
    \SetKwInput{KwInput}{Input}
    \SetKwInput{KwOutput}{Output}
    \SetKwFunction{EncodeToDevice}{encodeToDevice}
    
    \KwInput{Matrix $\mathbf{K}\in\mathbb{R}^{m\times n}$ on host}
    \KwOutput{Encoded block $\mathbf{M}_d \in\mathbb{R}^{(m+n)\times(m+n)}$ on accelerator}
    
    \BlankLine
    \KwStep~1: \textbf{Build symmetric block on host}\;
\quad $\mathbf{M}\gets
    \begin{bmatrix}
        \mathbf{0}_{m\times m} & \mathbf{K}\\
        \mathbf{K}^\top & \mathbf{0}_{n\times n}
    \end{bmatrix}$\;
\BlankLine
    \KwStep~2: \textbf{Encode once to accelerator}\;
    \quad $\mathbf{M}_d \gets \EncodeToDevice(\mathbf{M})$\;
\BlankLine
    \textbf{return} $\mathbf{M}_d$
\end{algorithm}

\vspace{3mm}
Together, these factors motivate the design of novel algorithms for in-memory computing such that algorithms may appear highly counterintuitive from the perspective of conventional CPU or GPU architectures, where the dominant sources of energy consumption and latency arise from entirely different bottlenecks.
We start the process with preparation for performing operator norm estimation~\cite{chambolle2011first}—that is, $\|\mathbf{K}\|_2$ given $\mathbf{K} \in \mathbb{R}^{m \times n}$ be the constraint matrix—and the PDHG method.
First, we construct a new symmetric block matrix, $\mathbf{M}$:


\begin{equation*}
    \mathbf{M} =
    \begin{bmatrix}
    \mathbf{0}_{m\times m} & \mathbf{K} \\
    \mathbf{K}^\top & \mathbf{0}_{n\times n}
    \end{bmatrix}
\end{equation*}

The matrix is encoded to the analog accelerator only once using Algorithm~\ref{alg:build-sym-block}. 
Unlike typical GPU implementations, this \emph{encode-once} strategy is critical for our system, as iteratively reprogramming both $\mathbf{K}$ and $\mathbf{K}^\top$ would be prohibitively costly on analog crossbars.
By employing $\mathbf{M}$, we can efficiently calculate the matrix-vector products required by our methods:

\begin{equation*}
    \mathbf{M} \begin{bmatrix} \mathbf{y} \\ \mathbf{x} \end{bmatrix} =
    \begin{bmatrix} \mathbf{K}\mathbf{x} \\ \mathbf{K}^\top\mathbf{y} \end{bmatrix}
\end{equation*}

This formulation will be used in both the step size estimation and PDHG implementation steps of the proposed method.

We describe our proposed method to solve linear optimization problems for in-memory device architectures. The first phase in our approach involves encoding the aggregated constraint matrix $\mathbf{M}$ on the RRAM as described in Algorithm \ref{alg:build-sym-block}.


\begin{figure*}[!htbp]
    \centering
    \includegraphics[width=\textwidth]{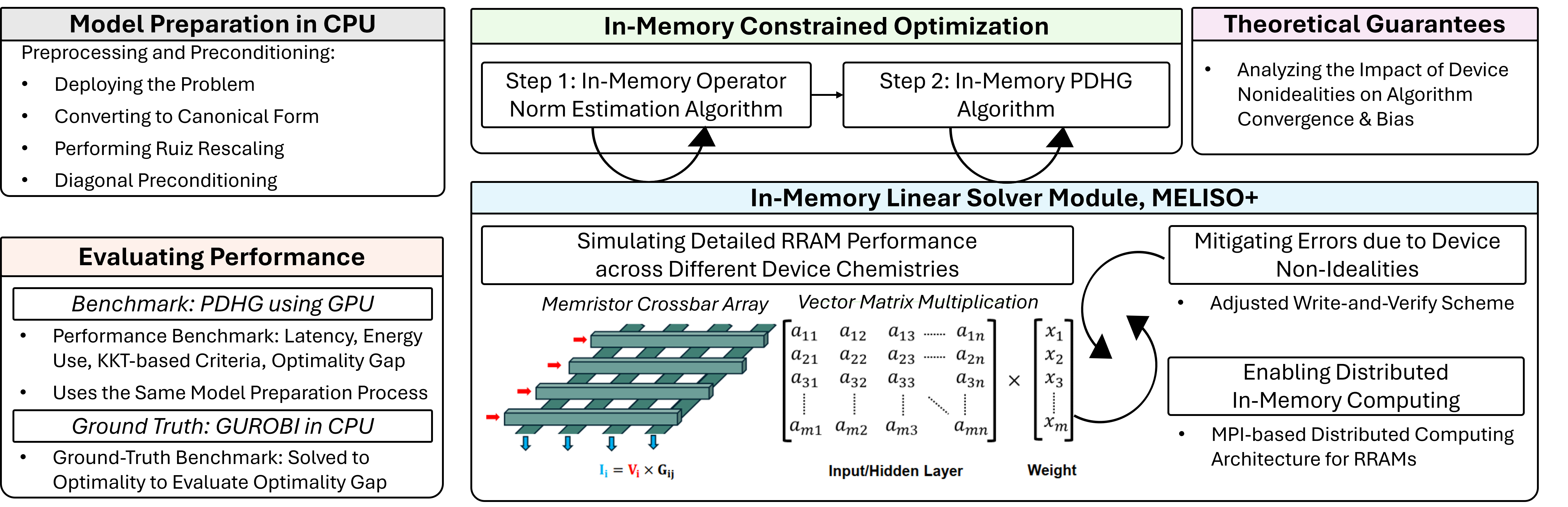}
    \caption{Overview of the GPU/MELISO+ framework for PDHG with matrix-vector multiplication (MVM) acceleration.}
\label{fig:overview}
\end{figure*}

Algorithm~\ref{alg:MVM-on-accelerators}, employs the symmetric block matrix $\mathbf{M}$ for all required MVMs to handle inputs for full-vector operations (in the operator norm estimation phase) and partial-vector operations (in the PDHG phase).

\begin{algorithm}[!htbp]
\caption{\textsc{matmulAccel}$(\mathbf{M}_d,\ \mathbf{u}_d,\ \texttt{mode})$}
\label{alg:MVM-on-accelerators}
\DontPrintSemicolon
\SetKwInput{KwInput}{Input}
\SetKwInput{KwOutput}{Output}
\SetKwFunction{Launch}{launchKernel}
\SetKwFunction{Length}{length}
\KwInput{Encoded $\mathbf{M}_d$; vector $\mathbf{u}_d$;
\texttt{mode} $\in \{\texttt{full},\texttt{A@x},\texttt{AT@y}\}$}
\KwOutput{Result vector(s)}
\BlankLine
\KwStep~1: \textbf{Pad input vector based on mode}\;
\uIf{\texttt{mode} = \texttt{full}}{
    $\mathbf{v}_d \gets \mathbf{u}_d$\;
}
\ElseIf{\texttt{mode} = \texttt{A@x}}{
    $\mathbf{v}_d \gets [\textsc{zeros}(m);\ \mathbf{u}_d]$\;
}
\Else{
    $\mathbf{v}_d \gets [\mathbf{u}_d;\ \textsc{zeros}(n)]$\;
}
\BlankLine
\KwStep~2: \textbf{Perform single device MVM}\;
\quad $\mathbf{w}_d \gets \Launch(\mathbf{M}_d,\ \mathbf{v}_d)$\;
\BlankLine
\KwStep~3: \textbf{Split and return appropriate output}\;
\quad $(\mathbf{t}_d,\ \mathbf{s}_d) \gets \big(\mathbf{w}_{d}[1{:}m],\ \mathbf{w}_{d}[m{+}1{:}m{+}n]\big)$\;
\uIf{\texttt{mode}=\texttt{full}}{\textbf{return} $\mathbf{w}_d$}
\uElseIf{\texttt{mode}=\texttt{A@x}}{\textbf{return} $\mathbf{t}_d$}
\uElse{\textbf{return} $\mathbf{s}_d$}
\end{algorithm}

\vspace{3mm}

Specifically, it first pads the input vector with zeros based on the selected \texttt{mode} (\texttt{full}, \texttt{A@x}, or \texttt{AT@y}).
Next, it executes a single MVM operation on the RRAM-based solver using an encoded-on-device symmetric matrix, $\mathbf{M}_d$.
Finally, the algorithm splits the resulting vector and returns the appropriate segment corresponding to the selected \texttt{mode}.
Moreover, to reduce the impact of device non-idealities, we employ the error-reduction techniques introduced in~\cite{vo2025harnessing} to ensure that $\mathbf{M}$ is encoded to the highest degree of accuracy. 

\subsection{Estimating Operator Norm via the Lanczos Method}
The convergence condition of PDHG requires coupled step sizes $\tau,\sigma>0$ such that $\tau\sigma \|\mathbf{K}\|_2^2 < 1$~\cite{chambolle2011first}. Since computing $\|\mathbf{K}\|_2$ exactly entails a full singular value decomposition (SVD), which is prohibitively expensive for large-scale problems, the \emph{power iteration} method is generally adopted to approximate the operator norm efficiently. Given a random initialization $\mathbf{v}^{\{0\}} \in \mathbb{R}^n$ with $\|\mathbf{v}^{\{0\}}\|_2=1$, we repeatedly apply the following updates:
\begin{equation}
    \mathbf{v}^{\{k+1\}} \;\leftarrow\;
\frac{\mathbf{K}^\top \mathbf{K}\,\mathbf{v}^{\{k\}}}
         {\|\mathbf{K}^\top \mathbf{K}\,\mathbf{v}^{\{k\}}\|_2},
    \qquad k = 0,1,2,\ldots.
\label{eq:power-iteration}
\end{equation}

As $k\to\infty$, the sequence $\{\mathbf{v}^{\{k\}}\}$ in~\eqref{eq:power-iteration} converges to the dominant eigenvector of
$\mathbf{K}^\top \mathbf{K}$, and the Rayleigh quotient,

\begin{equation*}
    \lambda^{\{k\}}
    = \frac{(\mathbf{v}^{\{k\}})^\top \mathbf{K}^\top \mathbf{K}\,\mathbf{v}^{\{k\}}}
           {(\mathbf{v}^{\{k\}})^\top \mathbf{v}^{\{k\}}}
\end{equation*}

approximates the largest eigenvalue of $\mathbf{K}^\top \mathbf{K}$. Therefore, the operator norm can be estimated as: $\|\mathbf{K}\|_2 \;\approx\; \sqrt{\lambda^{\{k\}}}$ after a modest number of iterations. In practice, a few iterations of \eqref{eq:power-iteration} are needed to obtain a reliable
estimate of $\|\mathbf{K}\|_2$, thereby enabling step-size selection through the condition $\eta < \|\mathbf{K}\|_2^{-1}$. This approach avoids costly matrix factorizations and ensures that the PDHG iterations remain stable even for large and sparse problem instances.

When implemented on an RRAM-based computing substrate, the products of iterative MVM are subject to inherent physical noise, despite best efforts to eliminate them. This stochasticity can corrupt the trajectory of the iterates, potentially degrading the quality of the eigenvector approximation and preventing the method from converging to a reliable estimate. To mitigate these effects and achieve more robust and rapid convergence, we therefore employ the \emph{Lanczos iteration}~\cite{doi:10.1137/1.9781421407944}. As a Krylov subspace method, the Lanczos algorithm iteratively constructs an orthonormal basis for the subspace spanned by the power iterates, $\{\mathbf{v}^{\{0\}}, (\mathbf{K}^\top\mathbf{K})\mathbf{v}^{\{0\}}, \dots, (\mathbf{K}^\top\mathbf{K})^{k-1}\mathbf{v}^{\{0\}}\}$.
The projection of $\mathbf{K}^\top\mathbf{K}$ onto this basis yields a small symmetric tridiagonal matrix, $\mathbf{T}_k \in \mathbb{R}^{k \times k}$, whose eigenvalues—commonly known as Ritz values, denoted by $v_{\text{Ritz}}$)—converge rapidly to the extreme eigenvalues of the original matrix.
The operator norm is then estimated as $\|\mathbf{K}\|_2 \approx \sqrt{\lambda_{\max}(\mathbf{T}_k)}$. By finding an optimal solution within this subspace, the Lanczos method provides a more stable and accurate estimate in fewer iterations, making it well-suited for our hardware implementation.

Since $\mathbf{M}$ is symmetric in our formulation , we can estimate its dominant eigenvalue, denoted by $\hat{\lambda}_{\max}(\mathbf{M})$, efficiently using the Lanczos iteration (Algorithm~\ref{alg:lanczos-svd}), where each Lanczos step consumes one full-vector MVM. We prove that the resulting estimate $\hat{\lambda}_{\max}(\mathbf{M})$ is in fact equivalent to the estimated operator norm $\sigma_{\max}(\mathbf{K})$ due to the following claim.

\begin{algorithm}[!htbp]
    \caption{\textsc{LanczosSVD}$(\mathbf{K},\ k_{\max},\ \epsilon)$}
    \label{alg:lanczos-svd}
    \DontPrintSemicolon
    \SetKwInput{KwInput}{Input}
    \SetKwInput{KwOutput}{Output}
    \SetKwFunction{BuildM}{buildSymBlock}
    \SetKwFunction{MVM}{matmulAccel}
    \SetKwFunction{Eigvalsh}{eigvalsh}
    
    \KwInput{Matrix $\mathbf{K}\in\mathbb{R}^{m\times n}$; max iterations $k_{\max}$; tolerance $\epsilon$}
    \KwOutput{Estimate of dominant singular value $\sigma_1$ of $\mathbf{K}$}
    
    \BlankLine
    \KwStep~1: \textbf{Setup}\;
    \quad $\mathbf{M}_d \gets \BuildM(\mathbf{K})$\;
    \quad Initialize $\mathbf{v}_0 \in \mathcal{N}(\mathbf{0},\mathbf{I}_{m+n})$, $\mathbf{v}_{-1} \gets \mathbf{0}$, $\beta_0 \gets 0$\;
    \quad $\mathbf{v}_0 \gets \mathbf{v}_0 / \|\mathbf{v}_0\|_2$\;
    \quad Initialize empty lists $\boldsymbol{\alpha}$, $\boldsymbol{\beta}$\;
    
    \BlankLine
    \KwStep~2: \textbf{Lanczos Iteration}\;
    \For{$j \gets 0$ \KwTo $k_{\max}-1$}{
        $\mathbf{w}_j \gets \MVM(\mathbf{M}_d, \mathbf{v}_j)$\;
    \BlankLine
        $\mathbf{w}_j \gets \mathbf{w}_j - \beta_j \mathbf{v}_{j-1}$\;
    \BlankLine
        $\alpha_j \gets \mathbf{v}_j^\top \mathbf{w}_j$\;
    \BlankLine
        $\mathbf{w}_j \gets \mathbf{w}_j - \alpha_j \mathbf{v}_j$\;
        $\beta_{j+1} \gets \|\mathbf{w}_j\|_2$\;
    
        \BlankLine
        Append $\alpha_j$ to $\boldsymbol{\alpha}$;
    Append $\beta_{j+1}$ to $\boldsymbol{\beta}$\;
    
        \BlankLine
        \If{$\beta_{j+1} < \epsilon$}{\textbf{break}\;}
    
        \BlankLine
        $\mathbf{v}_{j+1} \gets \mathbf{w}_j / \beta_{j+1}$\;
    }
    
    \BlankLine
    \BlankLine
    
    \KwStep~3: \textbf{Compute the Final Estimate }\;
    \quad Construct symmetric tridiagonal matrix $\mathbf{T}$ from $\boldsymbol{\alpha}$ and $\boldsymbol{\beta}$\;
    \quad $v_{\text{ritz}} \gets \Eigvalsh(\mathbf{T})$
    \quad $\sigma_1 \gets \max(|v_{\text{ritz}}|)$\;
    
    \BlankLine
    \textbf{return} $\sigma_1$
\end{algorithm}

\begin{proposition}
    The dominant eigenvalue of $\mathbf{M}$, $\lambda_{\max}(\mathbf{M})$, is equivalent to the dominant singular value of $\mathbf{K}$, $\sigma_{\max}(\mathbf{K})$. In other words, $\lambda_{\max}(\mathbf{M}) \equiv \sigma_{\max}(\mathbf{K})$.
\end{proposition}\label{thm:Lanczos}

\begin{proof}
    Proof is provided in Appendix \ref{subsubsec:proof_estimate_norm}.
\end{proof}

Leveraging on this formulation and result, we use Algorithm~\ref{alg:lanczos-svd} to obtain the theoretical PDHG step sizes to satisfy $\tau\sigma\|\mathbf{K}\|_2^2 < 1$, where $\|\mathbf{K}\|_2 \approx \hat{\sigma}_{\max}(\mathbf{K}) = \hat{\lambda}_{\max}(\mathbf{M})$.
Specifically, we use $\tau = \sigma = \eta / \hat{\sigma}_{\max}(\mathbf{K})$~\cite{applegate2023faster,lu2024cupdlpjlgpuimplementationrestarted, NEURIPS2021_a8fbbd3b}, with $\eta = 0.95$.

\subsection{Solving LPs via In-Memory PDHG}

\label{sec:PDHG_Elements_MELISO+}
The step size estimations are used to perform an enhanced PDHG outlined in Algorithm~\ref{alg:pdhg-enhanced}.
In each iteration, the PDHG executes exactly two MVMs on the accelerator, while all proximal operators and vector algebra remain on the host:

\begin{equation*}
    \mathbf{t}_d \leftarrow \mathbf{K}\mathbf{\bar{x}}^{\{k\}} \text{(dual step)};\quad \mathbf{s}_d \leftarrow \mathbf{K}^\top\mathbf{y}^{\{k+1\}} \text{(primal step)}
\end{equation*}
Since $\mathbf{M}$ is encoded once, there is neither $\mathbf{K}$ nor $\mathbf{K}^\top$ reprogramming overhead across iterations.

For adaptive step sizes, we employ the Nesterov momentum~\cite{nesterov1983method}. Compared to the approach proposed in~\cite{lu2024cupdlpjlgpuimplementationrestarted}, our method is simpler since:

\begin{itemize}
    \item In~\cite{lu2024cupdlpjlgpuimplementationrestarted}, a heuristic line search on the global step parameter is performed (with $\tau = \eta/\omega$, $\sigma = \eta\omega$). For every iteration $(k+1)$, one trial of the PDHG update is computed with the current $\eta$. Next, $\eta$ is accepted if it satisfies an inequality derived from the theoretical convergence rate $\mathcal{O}(1/k)$~\cite{lu2023unified}. Otherwise, $\eta$ is shrunk, and the trial is repeated. This potentially requires multiple tries per outer iteration when a trial is rejected, incurring extra MVMs.
    \item The $(\tau,\sigma)$ step sizes in our approach are adapted by a deterministic scalar factor $\theta_k \geq 0$. Thus, no extra MVMs beyond the usual two per iteration are introduced.
\end{itemize}

Convergence is evaluated by a lightweight, separate routine at the host level, stopping the method when either the KKT-based residuals or the relative differences in primal updates fall below prescribed tolerances.

Given a linear problem in its standard form, as described in Equation~\eqref{eq:standard-lp}, the Karush-Kuhn-Tucker (KKT) conditions at an optimal solution $(\mathbf{x}^\star,\mathbf{y}^\star)$ are given as:

\begin{align}
  &\text{Primal feasibility:}   &   \mathbf{K} \mathbf{x}^\star = \mathbf{b},\;\;
\mathbf{x}^\star \ge \mathbf{0}, \label{eq:kkt-primal} &\\
  &\text{Dual feasibility:}    &   \mathbf{K}^\top \mathbf{y}^\star \le \mathbf{c}, \label{eq:kkt-dual}\qquad &\\
  &\text{Complementarity:}    &  x^\star_i\!\left(c_i - (\mathbf{K}^\top \mathbf{y}^\star)_i\right) = 0 &,\ \forall i.
\label{eq:kkt-comp}
\end{align}

The KKT residuals directly measure deviation from \eqref{eq:kkt-primal}-\eqref{eq:kkt-comp} and are consequently natural, scale-aware stopping metrics for FOMs. Let $[\mathbf{v}]_+ := \max\{\mathbf{v}, \mathbf{0}\}$ applied element-wise. At the $(k{+}1)$-th iteration, we define the following residuals:

\begin{equation*}
     \begin{aligned}
      r_{\mathrm{pri}}  &:= \frac{\|\mathbf{K} \mathbf{x}^{\{k+1\}} - \mathbf{b}\|_2}{\,1 + \|\mathbf{b}\|_2\,} \\
      r_{\mathrm{dual}} &:= \frac{\|\mathbf{c} - \mathbf{K}^\top\mathbf{y}^{\{k+1\}} - \lambda\|_2}{\,1 + \|\mathbf{c}\|_2\,}, \\
      r_{\mathrm{iter}} &:= \frac{\|[\mathbf{x}^{\{k\}}-\mathbf{x}^{\{k+1\}}]_+\|_2}{\,1 + \|\mathbf{x}^{\{k+1\}}\|_2\,} \\
      r_{\mathrm{gap}}  &:= \frac{|\mathbf{c}^\top \mathbf{x}^{\{k+1\}} - \mathbf{K}^\top \mathbf{y}^{\{k+1\}}|}{\,1 + |\mathbf{c}^\top \mathbf{x}^{\{k+1\}}|
    + |\mathbf{K}^\top \mathbf{y}^{\{k+1\}}|\,}.
    \end{aligned}   
\end{equation*}

with $\lambda = [\mathbf{c} - \mathbf{K}^\top \mathbf{y}]_+$.

The denominators presented in our residuals implement a combined absolute and relative scaling, which is standard in modern first-order solvers~\cite{lu2024cupdlpjlgpuimplementationrestarted, NEURIPS2021_a8fbbd3b}. For a target tolerance $\varepsilon>0$ (we selected $\varepsilon=1e^{-6}$ in our study), we stop when the following condition is satisfied:

\begin{equation*}
  \max\{\, r_{\mathrm{pri}},\; r_{\mathrm{dual}},\; r_{\mathrm{iter}}\;,r_{\mathrm{gap}}\} \;\le\; \varepsilon.
\end{equation*}

\section{Theoretical Guarantees}\label{sec:theoretical_guarantees}
A critical concern when deploying iterative algorithms on analog hardware is the impact of device nonlinearities on theoretical convergence guarantees.
In our implementation of the RRAM-based solver, device nonlinearities inevitably introduce perturbations into the PDHG update equations, causing inexact primal-dual updates.
In other words, the analog MVMs associated with the Lanczos iteration and the primal-dual updates ($\mathbf{K}\mathbf{\bar{x}}^{\{k\}}$ and $\mathbf{K}^\top\mathbf{y}^{\{k\}}$) are subject to noise and drift caused by different sources~\cite{Gulafshan2025, 9761991}.
As a result, each iteration of these algorithms on analog accelerators effectively uses a perturbed operator. Consequentially, we formally express the inexact PDHG updates using the standard formulation of LPs, as described in Equation~\ref{eq:standard-lp}, onto to Equation~\ref{eq:pdhg-lp} as follows:
\begin{equation}
\boxed{
    \begin{aligned}
        \mathbf{x}^{\{k+1\}} &\leftarrow \operatorname{proj}_{\mathcal{X}}\!\bigl\{\mathbf{x}^{\{k\}} - \tau\,(\mathbf{c} - \mathbf{\widetilde{K}}_{\{k\}}^\top \mathbf{y}^{\{k\}})\bigr\},\\[3.0pt]
        \mathbf{y}^{\{k+1\}} &\leftarrow \operatorname{proj}_{\mathcal{Y}}\!\bigl\{\mathbf{y}^{\{k\}} + \sigma\Big[\mathbf{b} - \mathbf{\widetilde{K}}_{\{k\}}(2\mathbf{x}^{\{k+1\}}-\mathbf{x}^{\{k\}})\Big]\bigr\}.
    \end{aligned}
}
\label{eq:inexact-pdhg}
\end{equation}
where $\mathbf{\tilde{K}}_{\{k\}}^\top := \mathbf{K}^\top(1 + \pmb{\xi}^{\{k\}})$ and $\mathbf{\tilde{K}}_{\{k\}} := \mathbf{K}(1 + \pmb{\zeta}^{\{k\}})$ are perturbed matrices of $\mathbf{K}^\top$ and $\mathbf{K}$ observed at the $k$-th iteration, respectively.
These perturbations, $\{\pmb{\xi}^{\{k\}},\pmb{\zeta}^{\{k\}}\}$, can reflect random thermal noise, electronic noise, or systematic drift in the hardware. Although our results showcase that all the problems we solved achieved comparable accuracies with the GPU implementation of PDHG, it is still important to provide the theoretical guarantees for model performance. 

\subsection{Noisy In-Memory Lanczos}\label{subsubsec:noisy_inm_lanczos}
Through Proposition~1, we have established that $\lambda_{\max}(\mathbf{M}) = \sigma_{\max}(\mathbf{K})$. We assume that for every iteration $k$, the accumulated hardware noise is bounded, i.e., $||\mathbf{E}_k|| \leq \epsilon_{max}$. Thus, after $K$ iterations, we have: $||\mathbf{E}_{k|k = K}|| \leq K \epsilon_{\max}$,



\begin{lemma}[Perturbed Lanczos relation]\label{lem:app-lanczos-perturbed}
Consider Lanczos with full reorthogonalization on $\mathbf{M}$ with noisy matrix-vector products $\mathbf{w}_j = \mathbf{M}\mathbf{q}_j + \pmb{\zeta}_j$, where $\|\pmb{\zeta}_j\|\le \epsilon_{\max}$ for all $j$. Let $\mathbf{Q}_k^\top \mathbf{Q}_k = \mathbf{I}$ and let $\mathbf{\tilde{T}}_k$ denote the tridiagonal matrix computed by the RRAM-based accelerators, which is the perturbed version of the exact $\mathbf{T}_k$. Then, almost surely,
\[
\mathbf{M}\mathbf{Q}_k \;=\; \mathbf{Q}_k\,\mathbf{\tilde{T}}_k \;+\; \mathbf{r}_k \mathbf{e}_k^\top \;+\; \mathbf{E}_k,
\qquad \text{with} \quad \|\mathbf{E}_k\| \le k\,\epsilon_{\max},
\]
where $\mathbf{r}_k$ is the residual vector at the $k$-th iteration, and $\mathbf{E}_k$ is the accumulated perturbation due to hardware noise.
\end{lemma}
\begin{proof}
    Proof has been provided in Appendix \ref{subsec:proof-app-lanczos-perturbed}
\end{proof}

\begin{algorithm}[!htbp]
    \caption{\textsc{Enhanced-PDHG}}
    \label{alg:pdhg-enhanced}
    \DontPrintSemicolon
    \SetKwInput{KwInput}{Input}
    \SetKwInput{KwOutput}{Output}
    \SetKwFunction{RuizScale}{RuizRescaling}
    \SetKwFunction{DiagPre}{DiagonalPrecond}
    \SetKwFunction{Build}{buildSymBlock}
    \SetKwFunction{Matmul}{matmulAccel}
    \SetKwFunction{BoxProj}{proj}
    \SetKwFunction{SpecNorm}{LanczosSVD}
    \SetKw{KwStep}{Step}  

    \KwInput{Matrix $\mathbf{K}$; vector $\mathbf{b}$, $\mathbf{c}$; bounds $\mathbf{lb},\mathbf{ub}$;\\
    \text{ } max PDHG iterations $K$; tolerance $\varepsilon$; safety margin $\eta$;\\
    \text{ } Ruiz iterations $S$; max Lanczos iterations $J$;\\
    \text{ } Nesterov acceleration parameter $\gamma \geq 0$.}
    \KwOutput{Scaled-back solution $(\mathbf{x}_\mathrm{orig},\mathbf{y}_\mathrm{orig})$.}

    \BlankLine
    \KwStep~0: \textbf{Scaling and Preconditioning}\;
    \quad $(\mathbf{D}_1,\mathbf{D}_2) \leftarrow \RuizScale(\mathbf{K},S)$\;
    \quad $\tilde{\mathbf{K}} \leftarrow \mathbf{D}_1\,\mathbf{K}\,\mathbf{D}_2$;\quad
         $\tilde{\mathbf{b}} \leftarrow \mathbf{D}_1\,\mathbf{b}$;\quad
         $\tilde{\mathbf{c}} \leftarrow \mathbf{D}_2\,\mathbf{c}$\;
    \quad $\tilde{\mathbf{lb}} \leftarrow \mathbf{D}_2^{-1}\mathbf{lb}$;\quad
         $\tilde{\mathbf{ub}} \leftarrow \mathbf{D}_2^{-1}\mathbf{ub}$\;
    \quad $(\mathbf{T},\pmb{\Sigma}) \leftarrow \DiagPre(\tilde{\mathbf{K}})$\;
    \quad $\mathbf{M}_d \leftarrow \Build(\tilde{\mathbf{K}})$\;

    \BlankLine
    \KwStep~1: \textbf{Estimating Operator Norm}\;
    \quad $\rho \leftarrow \SpecNorm\!\big(\mathbf{M}_d,\,J\big)$\;

    \BlankLine
    \KwStep~2: \textbf{Initialization}\;
    \quad $\mathbf{x}^{\{0\}} \leftarrow \BoxProj\Big[\mathcal{N}(\mathbf{0},\mathbf{I}_{n \times n}),\,\tilde{\mathbf{lb}},\,\tilde{\mathbf{ub}}\Big]$;\quad
         $\mathbf{x}_{\mathrm{prev}} \leftarrow \mathbf{x}^{\{0\}}$\;
    \quad $\mathbf{y}^{\{0\}} \leftarrow \mathcal{N}(\mathbf{0},\mathbf{I}_{m \times m})$\;
    \quad $\tau \leftarrow \eta/\rho$;\quad $\sigma \leftarrow \eta/\rho$\;

    \BlankLine
    \KwStep~3: \textbf{Main loop}\;
    \For{$k=0$ \KwTo $K-1$}{
      $\theta_k \leftarrow 1/\sqrt{1 + 2 \gamma \tau}$\;
      $\tau \leftarrow \theta_k \tau$\;
      $\sigma \leftarrow \sigma/\theta_k$\;

      $\bar{\mathbf{x}}^{\{k\}} \leftarrow \mathbf{x}^{\{k\}} + \theta_k\big(\mathbf{x}^{\{k\}} - \mathbf{x}_{\mathrm{prev}}\big)$\;
      $\mathbf{v}_{\mathrm{bar}} \leftarrow \Matmul(\mathbf{M}_d,\,\bar{\mathbf{x}}^{\{k\}},\,\texttt{A@x})$\;
      $\mathbf{y}^{\{k+1\}} \leftarrow \mathbf{y}^{\{k\}} + \sigma\,\pmb{\Sigma}\,\big(\mathbf{v}_{\mathrm{bar}} - \tilde{\mathbf{M}}\big)$\;
      $\mathbf{x}_{\mathrm{prev}} \leftarrow \mathbf{x}^{\{k\}}$\;
      $\mathbf{u} \leftarrow \Matmul(\mathbf{M}_d,\,\mathbf{y}^{\{k+1\}},\,\texttt{A@Ty})$\;
      $\mathbf{g} \leftarrow \tilde{\mathbf{c}} + \mathbf{u}$\;
      $\mathbf{x}^{\{k+1\}} \leftarrow \BoxProj\!\big(\mathbf{x}^{\{k\}} - \tau\,\mathbf{T}\,\mathbf{g},\ \tilde{\mathbf{lb}},\ \tilde{\mathbf{ub}}\big)$\;

      \BlankLine
      \KwStep~4: \textbf{Convergence check}\;
      Compute $r_{\mathrm{pri}},\; r_{\mathrm{dual}},\; r_x,\; r_{\mathrm{gap}}$\;
      \If{$\max(r_{\mathrm{pri}},\; r_{\mathrm{dual}},\; r_x,\; r_{\mathrm{gap}}) < \varepsilon$}{\texttt{break}}
    }

    \BlankLine
    \textbf{return} $\big(\ \mathbf{x}_\mathrm{orig}=\mathbf{D}_2\,\mathbf{x}^{\{k+1\}},\ \ \mathbf{y}_\mathrm{orig}=\mathbf{D}_1^\top\,\mathbf{y}^{\{k+1\}} \ \big)$
\end{algorithm}

We leverage Lemma \ref{lem:app-lanczos-perturbed} in helping us derive the expected Ergodic bound for the noisy Ritz estimator for $L=\max_i \lambda_i(\mathbf{M})$. 

\begin{theorem}[Expected ergodic bound for noisy Ritz estimator of $L$]
\label{thm:app-lanczos-ritz}
Let $L=\max_i \lambda_i(\mathbf{M})$ denote the largest eigenvalue of $\mathbf{M}$, let $\theta_k = \lambda_{\max}(\mathbf{\tilde{T}}_k)$ denote the largest Ritz value at iteration $k$, and let $\bar\theta_K = \frac{1}{K}\sum_{k=1}^K \theta_k$ represent the averaged (ergodic) estimate. Then, there exist constants $C>0$ and $0 < \rho < 1$ such that
\[
\mathbb{E}_{\mathbf{q}_1,\pmb{\zeta}}\Big[|\theta_k - L|\Big]
\ \le\ C\rho^{(k-1)} + k\,\epsilon_{\max},
\]
\[
\mathbb{E}_{\mathbf{q}_1,\pmb{\zeta}}\Big[|\bar\theta_K - L|\Big]
\ \le\ \mathcal{O}\!\left(\tfrac{1}{K}\right) + \mathcal{O}(K\epsilon_{\max}).
\]
where the starting vector $\mathbf{q}_1$ is drawn uniformly at random from the unit sphere, $\rho = (\lambda_{2}/\lambda_1)$, represents the spectral ratio.
\end{theorem}
\begin{proof}
    Proof has been provided in Appendix \ref{subsec:proof-app-lanczos-ritz}.
\end{proof}

The theoretical result in Theorem \ref{thm:app-lanczos-ritz} captures the expected geometric convergence of the largest Ritz value under random initialization, with the convergence rate governed by the spectral ratio $\rho$. In Theorem \ref{thm:app-lanczos-ritz}, we assume that the algebraic multiplicity of eigenvalues of $\mathbf{M}$ is one, and the expectation is taken with respect to the distribution of the noise and the initial random vector. 

\subsection{Inexact In-Memory PDHG Updates}~\label{subsec:Inexact In-Memory PDHG Updates}
Next, we consider the inexact PDHG updates~\eqref{eq:inexact-pdhg} for the LP saddle form~\eqref{eq:pdhg-lp}, with RRAM-accelerated MVMs modeled as:
\begin{equation*}
\mathcal{M}(\mathbf{K};\mathbf{x}) = \mathbf{K}\mathbf{x} + \pmb{\zeta}, \qquad \mathcal{M}(\mathbf{K}^\top;\mathbf{y}) = \mathbf{K}^\top \mathbf{y} + \pmb{\xi}.
\end{equation*}

We show that even with additive noise in the updates, we can establish convergence bounds for PDHG under the following set of reasonable assumptions on the perturbations.
\begin{assumption}[Independence]
We assume that the noise sequence does not adversarially correlate with the iterates.
In other words, $\{\pmb{\xi}^{\{k\}},\pmb{\zeta}^{\{k\}}\}$ are independent random variables; $\{\pmb{\xi}^{\{k\}}$ is independent of $\{\mathbf{y}^{\{k\}},\mathbf{x}^{\{k\}}\}$; $\pmb{\zeta}^{\{k\}}$ is independent of $\{\mathbf{y}^{\{k\}},\mathbf{x}^{\{k\}},\mathbf{x}^{\{k+1\}}\}$.
\end{assumption}

\begin{assumption}[Unbiased perturbations]
We assume that $\widetilde{K}_{\{k\}}^\top$ and $\widetilde{K}_{\{k\}}$ are unbiased with respect to $\mathbf{K}^\top$ and $\mathbf{K}$, respectively.
In other words, given everything that has happened up to iteration $k$ (all past iterates and noise samples), the expected value of the next noise vector(s) is zero—that is, $\mathbb{E}[\pmb{\xi}^{\{k\}}] = 0$ and $\mathbb{E}[\pmb{\zeta}^{\{k\}}]= 0$.
\end{assumption}

\begin{assumption}[Bounded noise level]
We assume that the noise magnitude is uniformly bounded. In other words, $\|\pmb{\xi}^{\{k\}}\| \leq \delta $ and $\|\pmb{\zeta}^{\{k\}}\|\leq \delta$ for some fixed $\delta > 0$.
\end{assumption}

\begin{assumption}[Finite variance]
We assume that the noise variance is uniformly bounded. In other words, $\mathbb{E}\|\pmb{\xi}^{\{k\}}\|^2 \le \varphi_{\pmb{\xi}^{\{k\}}}^2$ and $\mathbb{E}\|\pmb{\zeta}^{\{k\}}\|^2 \le \varphi_{\pmb{\zeta}^{\{k\}}}^2$ for some fixed $\delta ^2 \geq \varphi_{\pmb{\xi}}^2>0$ and $\delta ^2 \geq \varphi_{\pmb{\zeta}}^2>0$.     
\end{assumption}

Under these assumptions, we establish the convergence of the robust PDHG method under a coupling rule driven by a noisy estimate. We first start with the robust step sizes from the noisy $\hat{L}$ as specified by Lemma \ref{lem:app-stepsize}.
\begin{lemma}[Safe coupling]\label{lem:app-stepsize}
Let $L=\|\mathbf{K}\|_2$ and suppose $|\hat{L}-L|\le \bar{\delta} L$, either in expectation or with high probability. If $\tau\sigma\ =\ \vartheta/\hat{L}^2$ given the safety margin, $\vartheta$, follows $\vartheta\in\bigl(0,(1-\bar{\delta})^2\bigr)$, then:
    $$
    \tau\sigma\,L^2\ \le\ \vartheta/(1-\bar{\delta})^2\ <\ 1
    $$
\end{lemma}
\begin{proof}
Proof provided in Appendix \ref{subsec:proof-app-stepsize}
\end{proof}

Lemma~\ref{lem:app-stepsize} provides a rule for safely choosing step sizes for the PDHG algorithm when the operator norm, $L$, is estimated with some error. By selecting step sizes based on the noisy estimate, $\hat{L}$, and a chosen safety margin $\vartheta$, it guarantees the standard PDHG convergence condition remains satisfied, ensuring the algorithm is sufficiently stable despite relying on an inexact norm calculation. 

We now establish Lemma \ref{lem:one-step-pdhg} to quantify the one-step inexact PDHG updates under assumptions of independent, unbiased and bounded perturbations of finite variance.

\begin{lemma}[One-step PDHG inequality with error terms]\label{lem:one-step-pdhg}
For any feasible solution $(\mathbf{x}, \mathbf{y})$, the iterates satisfy:
$$
\begin{aligned}
&\langle \mathbf{c}-\mathbf{K}^\top \mathbf{y}^{\{k+1\}},\mathbf{x}^{\{k+1\}}-\mathbf{x}\rangle\\
& \qquad+\langle \mathbf{K}(\mathbf{x}^{\{k+1\}})-\mathbf{b},\mathbf{y}^{\{k+1\}}-\mathbf{y}\rangle \leq \\
&\quad \frac{1}{2\tau}\left[\|\mathbf{x}-\mathbf{x}^{\{k\}}\|^2-\|\mathbf{x}-\mathbf{x}^{\{k+1\}}\|^2 - \|\mathbf{x}^{\{k\}} - \mathbf{x}^{\{k+1\}}\|^2\right] \\
&\quad +\frac{1}{2\sigma}\left[\|\mathbf{y}-\mathbf{y}^{\{k\}}\|^2-\|\mathbf{y}-\mathbf{y}^{\{k+1\}}\|^2 - \|\mathbf{y}^{\{k\}} - \mathbf{y}^{\{k+1\}}\|^2\right] \\
&\quad - \langle \mathbf{K}(\mathbf{x}^{\{k+1\}} - \mathbf{x}^{\{k\}}),\mathbf{y}^{\{k+1\}}-\mathbf{y}\rangle \\
&\quad +\langle \pmb{\xi}_k,\mathbf{x}-\mathbf{x}^{\{k+1\}}\rangle + \langle \pmb{\zeta}_k,\mathbf{y}-\mathbf{y}^{\{k+1\}}\rangle.
\end{aligned}
$$

\end{lemma}

\begin{proof}
Proof provided in Appendix \ref{subsec:proof-one-step-pdhg}
\end{proof}
Lemma~\ref{lem:one-step-pdhg} presents a one-step inequality that describes the progress of the PDHG algorithm in the presence of noise. It bounds the relationship between successive iterates and a feasible solution, similar to a standard PDHG analysis. Critically, it modifies the standard inequality to include explicit terms that account for the additive noise from the inexact matrix-vector products. We now establish the ergodic convergence of PDHG given our assumptions in Theorem \ref{thm:ergodic-pdhg}.

\begin{theorem}[Ergodic convergence of PDHG under noise assumptions]\label{thm:ergodic-pdhg}
Assume a saddle point $(\mathbf{x}^*, \mathbf{y}^*)$ exists for the LP, and the step sizes $\tau, \sigma > 0$ satisfying $\tau\sigma L^2 < 1$. Then, for the ergodic average $\bar{\mathbf{z}}_K = \frac{1}{K}\sum_{k=0}^{K-1} (\mathbf{x}^{\{k+1\}}, \mathbf{y}^{\{k+1\}})$, the expected primal-dual gap is bounded by:
$$
\mathbb{E}[\text{gap}(\bar{\mathbf{z}}_K)] \le \frac{C_0}{K} + \frac{\delta\sqrt{2}}{\sqrt{K}} + \text{constant},
$$
where $C_0$ depends on the initial distance to the optimal solution and $\delta$ is the noise bound.
\end{theorem}
\begin{proof}
Proof provided in Appendix \ref{subsec:proof-ergodic-pdhg}
\end{proof}
Theorem~\ref{thm:ergodic-pdhg} establishes the convergence rate for the PDHG algorithm when operating under the defined noise assumptions. It shows that the expected primal-dual gap of the averaged iterates decreases as the number of iterations, $K$, increases. The convergence bound consists of the standard $\mathcal{O}(1/K)$ rate found in noiseless analysis, an error term of $\mathcal{O}(1/\sqrt{K})$ that arises from the noise, and a constant term.

\section{Results}\label{sec:results}
We conduct a set of experiments to compare the performance of the proposed RRAM-based solver with a conventional GPU solver in terms of both latency and energy use. We first describe the experimental platform used to implement the algorithms and evaluate performance metrics, and then report results from solving a range of linear programming problems on RRAM devices with different chemistries and compare them with the corresponding implementations on GPUs.

\subsection{Experimental Setup}
\noindent\textit{Problem Instances:}
in our experiments, we employed several problems from the MIPLIB-2017 database~\cite{MIPLIB2017}. A detailed description of these problem instances is provided in Table \ref{tab:problems-fivecol}. These problems were specifically selected because their dimensions are well-suited to our target hardware architecture. The problem instances span a wide range of sizes, large enough to demonstrate the significant performance and energy-efficiency potential of RRAM-based accelerators for computationally intensive optimization tasks.

\begin{table}[!htbp]
    \centering
    \setlength{\tabcolsep}{1pt}
    \renewcommand{\arraystretch}{1.0}
    \caption{Problem types, sizes, objective values, and solve times.}
    \label{tab:problems-fivecol}
    \begin{tabular}{@{}%
    >{\centering\arraybackslash}p{.18\columnwidth}
    >{\centering\arraybackslash}p{.10\columnwidth}
    >{\centering\arraybackslash}p{.32\columnwidth}
    >{\centering\arraybackslash}p{.20\columnwidth}
    >{\centering\arraybackslash}p{.20\columnwidth}@{}}
    \toprule
    \textbf{Problem} & \textbf{Type} & \textbf{Size ($\mathbf{K}\in\mathbb{R}^{m\times n}$)} & \textbf{Objective Value} & \textbf{Solve Time (s)}\\
    \midrule
    gen-ip002 & Integer & $(24,41)$ & $\text{-}4783.7334$ & $10.52$ \\
    gen-ip016 & Integer & $(24,28)$ & $\text{-}9476.1552$ & $8.14$ \\
    gen-ip021 & Integer & $(28,35)$ & $2361.4542$ & $15.73$ \\
    gen-ip036 & Integer & $(46,29)$ & $\text{-}4606.6796$ & $21.98$ \\
    gen-ip054 & Integer & $(27,30)$ & $6840.9656$ & $12.45$ \\
    neos5 & Mixed & \makecell{$(402,253)$ \\ $n_b=53,\ n_c=10$} & $15.0000$ & $53.21$ \\
    assign1-5-8 & Mixed & \makecell{$(161,156)$ \\ $n_b=130,\ n_c=26$} & $212.0000$ & $58.66$ \\
    \bottomrule
    \end{tabular}
    \begin{minipage}{\columnwidth}
    \footnotesize
    Constraint matrix $\mathbf{K}\in\mathbb{R}^{m\times n}$, with $m$ constraints and $n$ variables.
    \\
    For mixed integer problems, $n_b$ and $n_c$ denote binary and continuous variables, respectively ($n=n_b+n_c$).
    \\
    \textbf{Solve Time} represents the duration of Gurobi's process for solving the relaxed problem.
    \end{minipage}
\end{table}

\noindent\textit{Hardware Platforms:} The RRAM-based accelerators were simulated based on a system architecture consisting of a 4-by-4 array of 64-by-64 RRAM crossbars. This configuration provides a total logical dimension of 256-by-256 for matrix operations, which can natively accommodate the constraint matrices of the selected benchmarks. The performance characteristics, where time-energy complexity is dominated by internal read and write operations, were quantified using our full-stack simulation module, MELISO+~\cite{10766540,vo2025harnessing}. The simulation models two distinct, promising RRAM technologies: EpiRAM~\cite{c30a52423ab14e6c83bcac6aa8cae784}, and TaO$_x$-HfO$_x$~\cite{8510690}. As demonstrated in~\cite{10766540}, these materials exhibit properties consistent with the noise assumptions outlined in Section~\ref{subsec:Inexact In-Memory PDHG Updates}.. The simulation outcomes provide a detailed analysis of model accuracy, energy consumption, and latency patterns throughout the implementation of the algorithm.

The baseline GPU implementation was executed on a single compute node of the Oklahoma State University (OSU) Pete Supercomputer. The compute node has the following specification: One (1) Intel Xeon Gold 6130 Processor (CPU), two (2) NVIDIA Quadro RTX6000 (GPU), and 192GB RAM (Memory). Only a single GPU was used to accelerate all matrix-vector multiplication (MVM) operations within the PDHG and Lanczos algorithms. We measured both computation and data movement costs using the Zeus framework~\cite{zeus-nsdi23} to ensure a comprehensive performance profile.

\noindent\textit{Solution Accuracy:} In all our experiments, we quantify the solution accuracy by reporting the relative error, defined as:

\begin{equation}
    \Delta_{\text{rel}}:= \frac{|z - z^*|}{|z|}
\end{equation}

where $z^*$ are the results produced by the solver under evaluation (RRAM-based or GPU) and $z$ denotes the ground-truth value. For the Lanczos iteration, $\Delta_{\text{rel}}$ is benchmarked against the \texttt{np.linalg.svd()} provided by the Numpy Python package. For PDHG, the relative error, $\Delta_{\text{rel}}$ is benchmarked against the objective value of the problems obtained from the commercial solver Gurobi~\cite{gurobi}.

\begin{table*}[h]
\centering
\setlength{\tabcolsep}{3pt} 
\renewcommand{\arraystretch}{1.1} 
\caption{Comparative overview of optimality gap, total energy consumption, and total latency across all accelerator platforms for the Lanczos and PDHG methods.}
\label{tab:comparative-overview-combined}
\begin{tabular}{l l c cc cc c cc cc}
\toprule
& & \multicolumn{5}{c}{\textbf{\textit{Step 1:} Norm Estimation {via Lanczos} }} & \multicolumn{5}{c}{\textbf{\textbf{\textit{Step 2:} PDHG}}} \\
\cmidrule(lr){3-7} \cmidrule(lr){8-12}
\multirow{2}{*}{\textbf{Problem}} & \multirow{2}{*}{\textbf{Solver}} & \textbf{Gap} & \multicolumn{2}{c}{\textbf{Energy [J]}} & \multicolumn{2}{c}{\textbf{Latency [s]}} & \textbf{Gap} & \multicolumn{2}{c}{\textbf{Energy [J]}} & \multicolumn{2}{c}{\textbf{Latency [s]}} \\
\cmidrule(lr){4-5}\cmidrule(lr){6-7} \cmidrule(lr){9-10}\cmidrule(lr){11-12}
& & ($\Delta_{\text{rel}}$) & \textbf{Total} & \textbf{Factor} & \textbf{Total} & \textbf{Factor} & ($\Delta_{\text{rel}}$) & \textbf{Total} & \textbf{Factor} & \textbf{Total} & \textbf{Factor} \\
\midrule
\multirow{3}{*}{gen-ip002} & gpuPDLP & $---$ & 11.9356 & $---$ & 0.3398 & $---$ & $1.52\text{E}^{-5}$ & 836.7268 & $---$ & 69.2601 & $---$ \\
 & EpiRAM & $1.30\text{E}^{-3}$ & 0.7542 & $15.83\times$ & 0.3354 & $1.01\times$ & $5.02\text{E}^{-6}$ & 18.1730 & $46.04\times$ & 3.5964 & $19.26\times$ \\
 & TaO$_x$-HfO$_x$ & $2.00\text{E}^{-4}$ & 0.0136 & \textbf{877.62$\times$} & 0.0644 & $5.28\times$ & $6.01\text{E}^{-6}$ & 6.5670 & \textbf{127.42$\times$} & 3.8090 & $18.18\times$ \\
\midrule
\multirow{3}{*}{gen-ip016} & gpuPDLP & $---$ & 5.2834 & $---$ & 0.1442 & $---$ & $4.69\text{E}^{-5}$ & 2811.5477 & $---$ & 141.9507 & $---$ \\
 & EpiRAM & $4.10\text{E}^{-3}$ & 2.7538 & $1.92\times$ & 1.3809 & $0.10\times$ & $1.18\text{E}^{-4}$ & 20.5487 & \textbf{136.82$\times$} & 3.7760 & $37.59\times$ \\
 & TaO$_x$-HfO$_x$ & $7.00\text{E}^{-4}$ & 0.0661 & $79.93\times$ & 0.0353 & $4.09\times$ & $1.27\text{E}^{-4}$ & 7.6441 & \textbf{367.80$\times$} & 4.4653 & $31.79\times$ \\
\midrule
\multirow{3}{*}{gen-ip021} & gpuPDLP & $---$ & 4.4468 & $---$ & 0.1219 & $---$ & $2.64\text{E}^{-6}$ & 172.1724 & $---$ & 6.0919 & $---$ \\
 & EpiRAM & $1.80\text{E}^{-3}$ & 1.0060 & $4.42\times$ & 0.1491 & $0.82\times$ & $2.98\text{E}^{-2}$ & 20.8827 & $8.24\times$ & 6.2145 & $0.98\times$ \\
 & TaO$_x$-HfO$_x$ & $1.92\text{E}^{-2}$ & 0.0169 & \textbf{263.12$\times$} & 0.0811 & $1.50\times$ & $2.98\text{E}^{-2}$ & 4.9248 & $34.96\times$ & 2.8739 & $2.12\times$ \\
\midrule
\multirow{3}{*}{gen-ip036} & gpuPDLP & $---$ & 2.8351 & $---$ & 0.1090 & $---$ & $4.77\text{E}^{-6}$ & 365.0960 & $---$ & 8.2201 & $---$ \\
 & EpiRAM & $4.00\text{E}^{-4}$ & 18.4238 & $0.15\times$ & 7.1263 & $0.02\times$ & $5.53\text{E}^{-5}$ & 19.8905 & $18.35\times$ & 7.2341 & $1.14\times$ \\
 & TaO$_x$-HfO$_x$ & $1.70\text{E}^{-3}$ & 0.0762 & $37.21\times$ & 0.1400 & $0.78\times$ & $5.35\text{E}^{-5}$ & 0.3862 & \textbf{945.36$\times$} & 3.6986 & $2.22\times$ \\
\midrule
\multirow{3}{*}{gen-ip054} & gpuPDLP & $---$ & 2.7225 & $---$ & 0.1103 & $---$ & $4.30\text{E}^{-5}$ & 4870.1550 & $---$ & 244.2323 & $---$ \\
 & EpiRAM & $4.00\text{E}^{-4}$ & 4.3371 & $0.63\times$ & 3.6241 & $0.03\times$ & $3.36\text{E}^{-3}$ & 17.9449 & \textbf{271.40$\times$} & 19.2164 & $12.71\times$ \\
 & TaO$_x$-HfO$_x$ & $6.30\text{E}^{-3}$ & 0.0626 & $43.49\times$ & 0.0352 & $3.13\times$ & $6.98\text{E}^{-3}$ & 7.4425 & \textbf{654.36$\times$} & 4.3460 & $56.20\times$ \\
\midrule
\multirow{3}{*}{neos5} & gpuPDLP & $---$ & 2.7500 & $---$ & 0.0895 & $---$ & $5.64\text{E}^{-4}$ & 269.4557 & $---$ & 10.1614 & $---$ \\
 & EpiRAM & $0.00$ & 0.1718 & $16.01\times$ & 0.0846 & $1.06\times$ & $7.69\text{E}^{-3}$ & 0.3508 & \textbf{768.12$\times$} & 0.1899 & $53.51\times$ \\
 & TaO$_x$-HfO$_x$ & $0.00$ & 0.0962 & $28.59\times$ & 0.0433 & $2.07\times$ & $1.48\text{E}^{-2}$ & 0.0539 & \textbf{4999.18$\times$} & 0.0314 & \textbf{323.61$\times$} \\
\midrule
\multirow{3}{*}{assign1-5-8} & gpuPDLP & $---$ & 2.8117 & $---$ & 0.1006 & $---$ & $3.81\text{E}^{-5}$ & 1417.4179 & $---$ & 410.2391 & $---$ \\
 & EpiRAM & $2.80\text{E}^{-3}$ & 1.0730 & $2.62\times$ & 0.7870 & $0.13\times$ & $1.70\text{E}^{-3}$ & 6.5448 & \textbf{768.12$\times$} & 3.7870 & \textbf{108.33$\times$} \\
 & TaO$_x$-HfO$_x$ & $1.00\text{E}^{-4}$ & 0.5353 & $5.25\times$ & 0.1883 & $0.53\times$ & $4.23\text{E}^{-3}$ & 3.7889 & \textbf{374.11$\times$} & 2.0852 & \textbf{196.74$\times$} \\
\bottomrule
\end{tabular}
\end{table*}

\subsection{Energy and Latency Analysis}
\vspace{-1mm}
Table~\ref{tab:comparative-overview-combined} presents a comprehensive performance comparison of the Lanczos and PDHG algorithms on three different hardware accelerators for MVMs: a traditional GPU implementation of PDHG called \texttt{gpuPDLP}, EpiRAM \cite{c30a52423ab14e6c83bcac6aa8cae784}, and TaO$_x$-HfO$x$ \cite{8510690}. For the RRAM-based implementation, the energy and latency savings are quantified as a multiplicative factor (the {Factor} column) relative to the GPU's performance, providing a direct measure of potential efficiency improvement for switching from conventional GPU to the proposed RRAM-based method. Across all test cases, the negligible relative errors confirm that all three platforms reliably converge to high-quality results. We provide a detailed breakdown of measured energy and latency components in Section~\ref{sec:detailed_results} of the Appendix.

The results in Table~\ref{tab:comparative-overview-combined} demonstrate the profound potential of RRAM-based accelerators for energy-efficient computing. This advantage is most pronounced for the PDHG method, where both EpiRAM and TaO$_x$-HfO$_x$ achieve massive reductions in energy consumption, regularly outperforming the GPU by two to three orders of magnitude—with savings reaching nearly 5000$\times$ for the TaO$_x$-HfO$_x$ platform. While running the Lanczos iteration, the RRAM-based platforms continue to demonstrate exceptional energy efficiency. The TaO$_x$-HfO$_x$ accelerator, in particular, reduces energy consumption dramatically compared to the GPU, with savings ranging from approximately 29$\times$ to a remarkable 878$\times$. EpiRAM also shows considerable energy savings for the Lanczos method, though its performance is less consistent across all problem instances.

The analysis of latency demonstrates significant advantages for the proposed IMC method, but also reveals interesting insights into performance trade-offs that are dependent on the algorithm. The speedups for the Lanczos iteration are considerable. The TaO$_x$-HfO$_x$ platform remains highly competitive, outperforming the GPU in most cases with a speedup of up to 5.3$\times$, but the EpiRAM platform generally exhibits higher latency than both the GPU and TaO$_x$-HfO$_x$ for this algorithm. However, the Lanczos iterations account for only a small fraction of the overall latency and energy consumption, e.g. $1-4\%$ and $0.05-9\%$ respectively for the gen-ip002 problem. For the PDHG method, which takes the lion's share of energy use and latency, the savings are significantly more pronouced. Both RRAM-based solvers significantly outperform the GPU—with EpiRAM and TaO$_x$-HfO$_x$ being up to 108$\times$ and 323$\times$ faster, respectively. In fact, these improvements become more prominent as the problem sizes increase, suggesting further benefits in high-scale applications.

All these improvements have been summarized in Table \ref{tab:improvement} for the entire process (i.e., including Steps 1 and 2) in order to demonstrate the total overall improvement in energy consumption and latency in RRAM-based systems compared to equivalent GPU implementations. The results show a clear advantage of using an RRAM-based IMC solver. Among the two RRAM technologies evaluated, the TaO\textsubscript{x}-HfO\textsubscript{x}-based solver consistently delivered superior performance. This improvement is primarily attributed to the physics of the device, which allows lower voltage and duration of programming signals required to adjust the conductance state of the memory cell \cite{8510690}, which in turn enables more efficient analog matrix-vector operations.

\begin{table}[!htbp]
\centering
\caption{Overall Energy \& Latency Improvement of RRAM Solvers over GPU}
\label{tab:improvement}
\begin{tabular}{l cc cc}
\toprule
& \multicolumn{4}{c}{Approx. Improvement Factor over GPU} \\
\cmidrule(lr){2-5}
\textbf{Problem} & \multicolumn{2}{c}{EpiRAM} & \multicolumn{2}{c}{TaO\textsubscript{x}-HfO\textsubscript{x}} \\
\cmidrule(lr){2-3} \cmidrule(lr){4-5}
& Power & Latency & Power & Latency \\
\midrule
gen-ip002     & 45$\times$  & 18$\times$     & 129$\times$  & 18$\times$     \\
gen-ip016     & 121$\times$ & 28$\times$     & 365$\times$  & 32$\times$     \\
gen-ip021     & 8$\times$   & 1$\times$     & 36$\times$   & 2$\times$     \\
gen-ip036     & 10$\times$  & 1$\times$     & 796$\times$  & 2$\times$     \\
gen-ip054     & 219$\times$ & 11$\times$     & 649$\times$  & 56$\times$     \\
neos5         & 521$\times$ & 37$\times$     & 1813$\times$ & 137$\times$     \\
assign1-5-8   & 186$\times$ & 90$\times$     & 328$\times$  & 180$\times$     \\
\bottomrule
\end{tabular}
\end{table}

\begin{figure*}[!htbp]
  \centering
  \begin{subfigure}{0.32\linewidth}
    \centering
    \includegraphics[width=\linewidth]{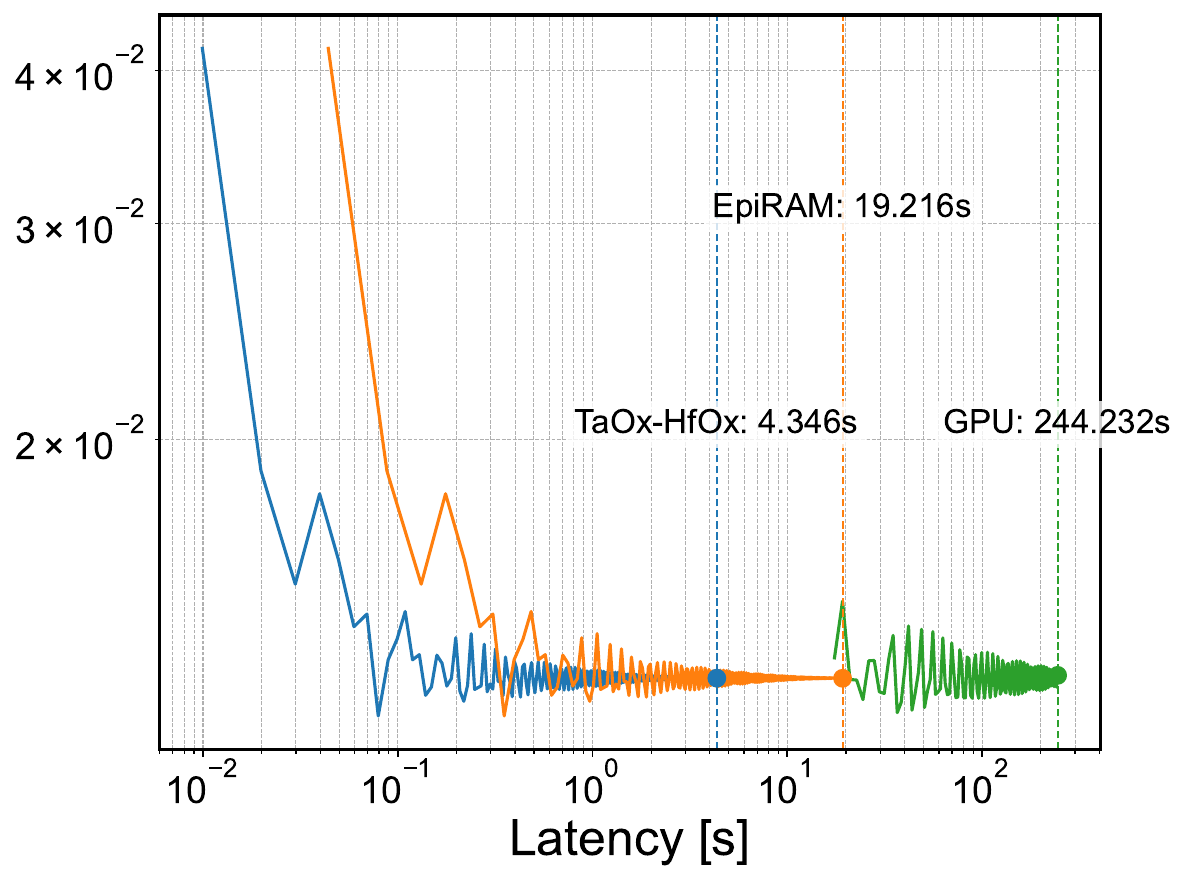}
    \caption{Primal residuals.}
    \label{fig:convergence_rpri}
  \end{subfigure}
  \hfill
  \begin{subfigure}{0.32\linewidth}
    \centering
    \includegraphics[width=\linewidth]{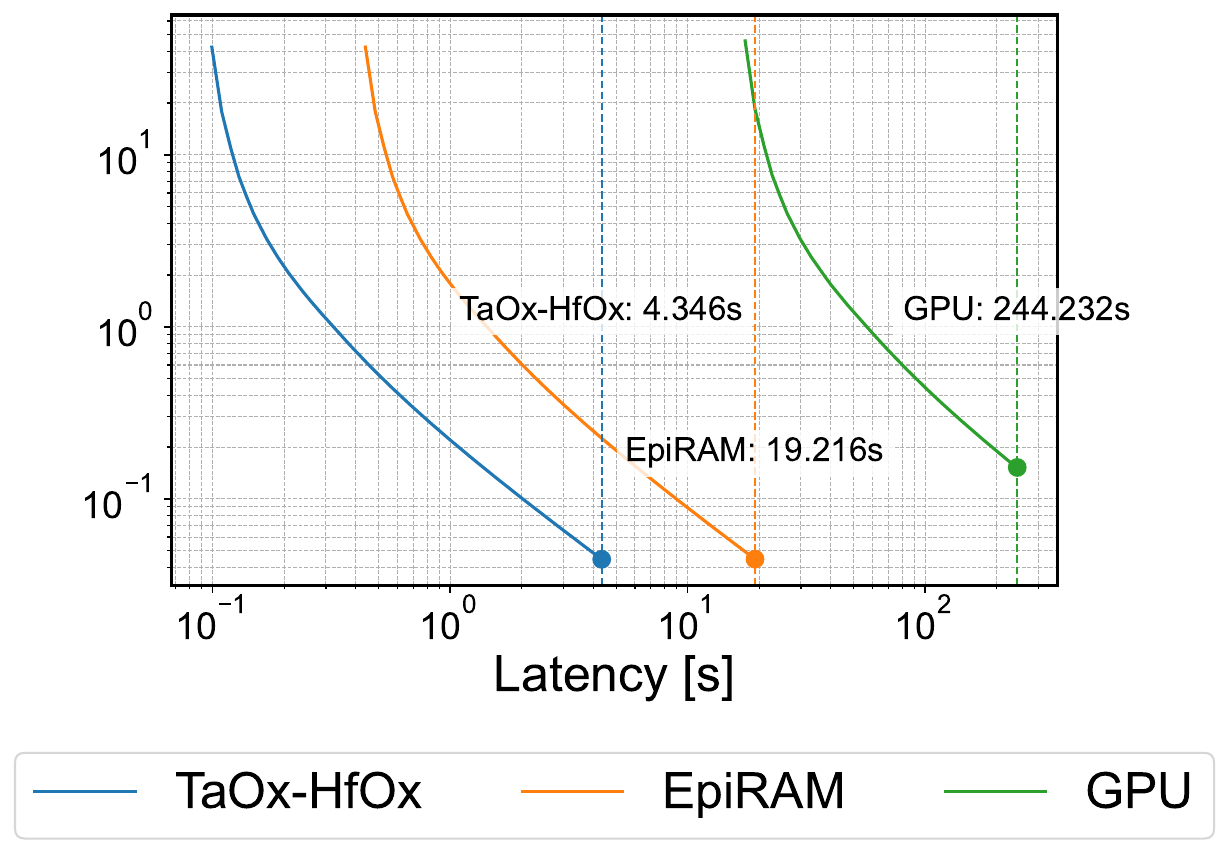}
    \caption{Dual residuals.}
    \label{fig:convergence_rdual}
  \end{subfigure}
  \hfill
  \begin{subfigure}{0.32\linewidth}
    \centering
    \includegraphics[width=\linewidth]{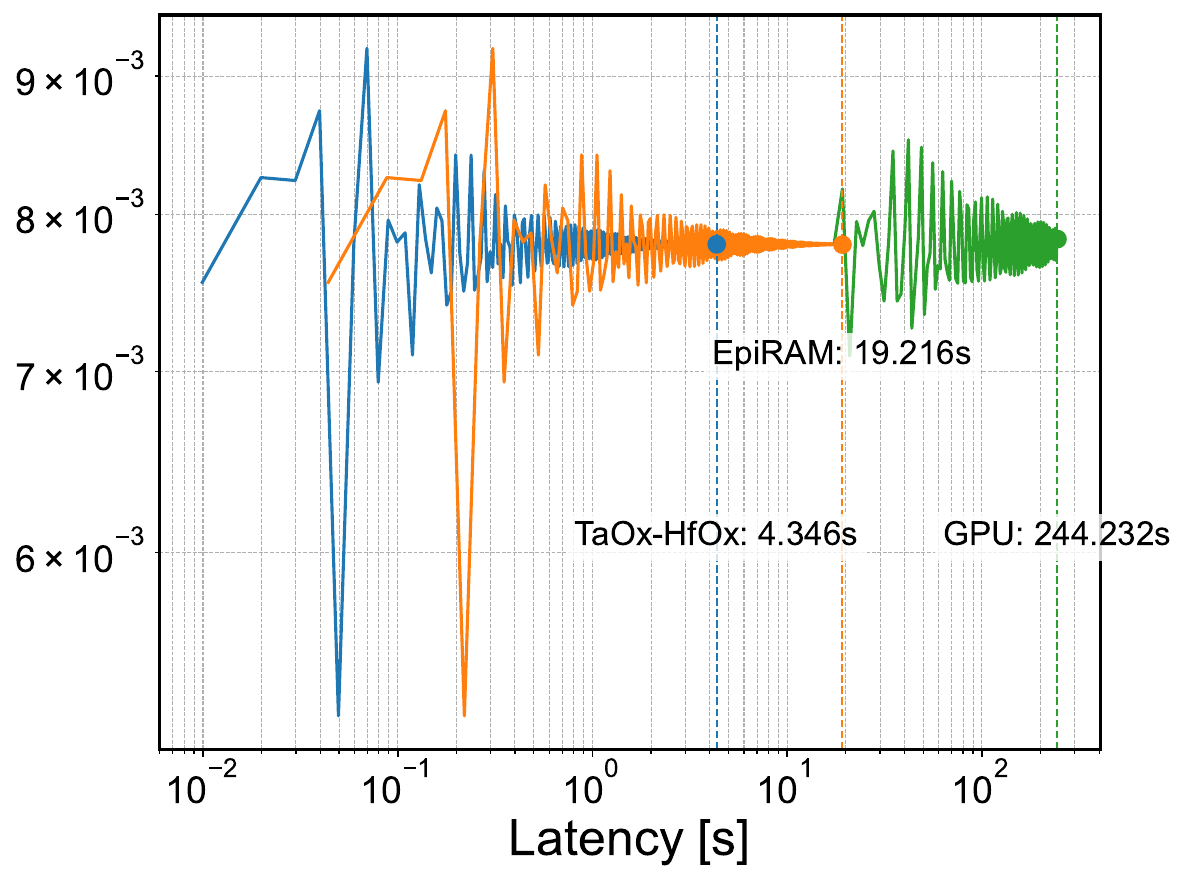}
    \caption{Optimality gap.}
    \label{fig:convergence_optgap}
  \end{subfigure}
  \caption{Convergence and accuracy on the gen-ip054 instance across two RRAM-based frameworks (EpiRAM and TaO$_x$/HfO$_x$) for the PDHG method.}
  \label{fig:convergence_latency}
\end{figure*}

\subsection{Convergence versus Latency Analysis}
We next study the progression of different model and performance parameters over time. We select gen-ip054 as the problem instance to conduct our convergence analysis of the PDHG method on the basis of latency complexities. Unlike other instances, it yields substantially higher PDHG optimality gaps on RRAM compared to GPU, indicating nontrivial convergence dynamics that make KKT residual trajectories informative.

Figure~\ref{fig:convergence_latency} provides a comparison on the convergence (in terms of primal-dual residuals and optimality gap) of two RRAM-based accelerators and GPU implementation over their respective latencies. We observe that, on gen-ip054, PDHG running on TaO$_x$-HfO$_x$ reaches the target stopping tolerance faster (approx.~4.35 s) than that of EpiRAM (approx.~19.22 s) under the same algorithmic settings and stopping criteria. This approx.~4.4$\times$ wall-clock advantage reflects lower per-iteration overheads and higher effective throughput for matrix-vector operations TaO$_x$-HfO$_x$.

Furthermore, the TaO$_x$-HfO$_x$ trace exhibits a steeper initial decrease in relative optimality gap 
followed by a consistent, nearly linear decline. In contrast, EpiRAM shows a longer transient/plateau before entering a similar steady-decay regime, suggesting additional fixed costs that amortize more slowly over the early iterations. Therefore, for workloads similar to gen-ip054—where a moderate optimality gap yields acceptable solution quality—TaO$_x$-HfO$_x$ provides a highly competitive time-to-solution for any target accuracy, while EpiRAM remains viable when energy per solve is prioritized and wall-clock is less critical.


\section{Discussion}
In this paper, we introduced and evaluated a novel in-memory computing framework designed to solve large-scale constrained optimization problems on RRAM-based hardware. Our approach marries a state-of-the-art first-order method, the PDHG algorithm, with a hardware-aware co-design that leverages the unique characteristics of analog RRAM devices.

A cornerstone of our framework is the deep algorithmic-hardware co-design—a principle that moves beyond simply implementing an algorithm in its conventional form on a new accelerator. Instead, we fundamentally restructure the algorithm to align with the physical realities of the hardware. This is especially critical for the architecture we are considering—a distributed 4-by-4 array of 64-by-64 RRAM crossbars, where a naive implementation could easily reintroduce much greater bottlenecks than those we aim to eliminate.

The challenge with a distributed system of crossbars is that while computation within each crossbar is fast, communication between them can become a performance limiter. Our approach directly addresses this through the proposed symmetric block-matrix formulation ($\mathbf{M}$). Here is how this represents a cornerstone of the design:
\begin{enumerate}[leftmargin=*]
    \item \textit{Single, Unified Mapping}: Instead of treating the forward ($\mathbf{K}\mathbf{x}$) and transpose ($\mathbf{K}^\top\mathbf{y}$) operations as separate tasks requiring different data layouts, we construct the larger $\mathbf{M}$ matrix directly on the host CPU. This unified matrix is then partitioned and mapped across the 16 physical crossbars, where it is encoded once using the write-and-verify scheme~\cite{vo2025harnessing}. This strategy is significant as it amortizes the high initial cost of writing to the RRAM devices over the entire iterative process of the Lanczos and PDHG algorithms.
    \item \textit{Elimination of Iterative Communication Overhead}: Once $\mathbf{M}$ is encoded across the distributed array, the system effectively acts as a single, large logical accelerator. During the PDHG iterations, the input vectors are broadcasted to all crossbars simultaneously. Each crossbar performs its local portion of the MVM in parallel. Resulting analog output currents are then aggregated to produce the final output vector. Crucially, no matrix reprogramming or data shuffling between crossbars is needed during the iterative solve, which significantly reduces energy use and latency.
    \item \textit{Abstraction and Simplicity:} Our \texttt{matmulAccel} function (Algorithm~\ref{alg:MVM-on-accelerators}) provides a clean abstraction layer. It adaptively handles the necessary padding and slicing of vectors to ensure the correct MVM ($\mathbf{K}\mathbf{x}$ or $\mathbf{K}^\top\mathbf{y}$) is performed on the distributed $\mathbf{M}$ matrix. This allows the main algorithms (Algorithm~\ref{alg:lanczos-svd} and~\ref{alg:pdhg-enhanced}) to remain straightforward, invoking what appears to be a simple MVM call, while the underlying distributed hardware execution is managed via a complex set of rules that can generalize to all matrix types.
\end{enumerate}
In contrast, a naive approach without this co-design might involve encoding $\mathbf{K}$ and $\mathbf{K}^\top$ on separate sets of crossbars, doubling the programming overhead and requiring logic to switch between hardware resources at each step. By creating a static data mapping that serves all computational needs of the algorithm, our framework minimizes on-chip communication and maximizes the parallel potential of the distributed RRAM architecture.

Furthermore, our framework successfully navigates the challenge of analog noise and device non-idealities. By employing the more robust Lanczos iteration for operator norm estimation and providing theoretical guarantees for the convergence of PDHG under noisy updates, we establish that RRAM-based IMC is a viable and reliable platform for high-precision iterative algorithms. The consistent convergence across all benchmark problems validates this central claim. These claims are further supported by the theoretical guarantees provided in this work.

\section{Conclusion}
In this paper, we establish the significant potential of RRAM-based in-memory computing for solving linear optimization problems through iterative PDHG. We demonstrate the applicability of in-memory PDHG through a full-stack framework that successfully solves canonical linear programs on a simulated RRAM architecture. The in-memory framework achieves optimality gaps comparable to those from commercial solvers and GPU-accelerated methods. Moreover, the proposed framework yields orders-of-magnitude improvements in performance and efficiency. For the iterative PDHG method, our results show energy and latency reductions of up to three orders of magnitude, respectively, compared to a GPU baseline. Our co-designed PDHG algorithm, featuring a symmetric block-matrix formulation and robust operator norm estimation, is fundamentally suited for the physical properties of RRAM devices. This approach minimizes expensive write operations and ensures convergence despite inherent analog hardware noise. Theoretical analysis provides robust convergence guarantees for the PDHG method in the presence of inexact updates, solidifying the mathematical foundation for using iterative methods on noisy analog computing platforms.

This research opens several promising avenues for future investigation. 
The current benchmarks, while substantial, are sized to fit within a single logical crossbar array. Future work will explore scalable algorithms for problems that require virtualization and partitioning across multiple RRAM arrays, building upon our previous work~\cite{vo2025harnessing}. Finally, this framework's core reliance on MVMs makes it adaptable to a broader class of optimization and machine learning/deep learning problems, which represent exciting directions for expanding the application of energy-efficient in-memory computing.



\section{Acknowledgements}
This research was supported in part by the National Science Foundation Secure and Trustworthy Cyberspace (SaTC) program under Award No. 2348411. Additionally, some of the computing for this work was performed at the High Performance Computing Center at Oklahoma State University, supported in part through the National Science Foundation Grant OAC-1531128.

\bibliographystyle{unsrt}
\bibliography{main}

\clearpage

\appendix
\section{Appendix}
\subsection{Detailed Results}\label{sec:detailed_results}
Tables 4 and 5 provide a detailed performance breakdown of the Lanczos and PDHG algorithms on three different hardware platforms: a GPU (\texttt{gpuPDLP}), and two simulated RRAM-based accelerators (EpiRAM and TaO$_x$–HfO$_x$). For each benchmark problem, the tables report the total energy consumed (in Joules) and the total latency (latency in seconds), further broken down into constituent operations including data transfer (Host-to-Device and Device-to-Host for the GPU), computation (Solve), and memory operations (Write/Read for the RRAM devices). 

Table 4 focuses on the initial operator norm estimation phase using the Lanczos iteration, showing the estimated dominant singular value, $\hat{\sigma}_{\max}(\mathbf{M})$, and the number of iterations to convergence ($k^\star$). 
Table 5 details the performance of the main PDHG algorithm, comparing the final objective values achieved by each platform against the optimal value obtained from the Gurobi solver on relaxed problems.

\begin{figure*}[!htbp]
\centering
\includegraphics[width=\linewidth]{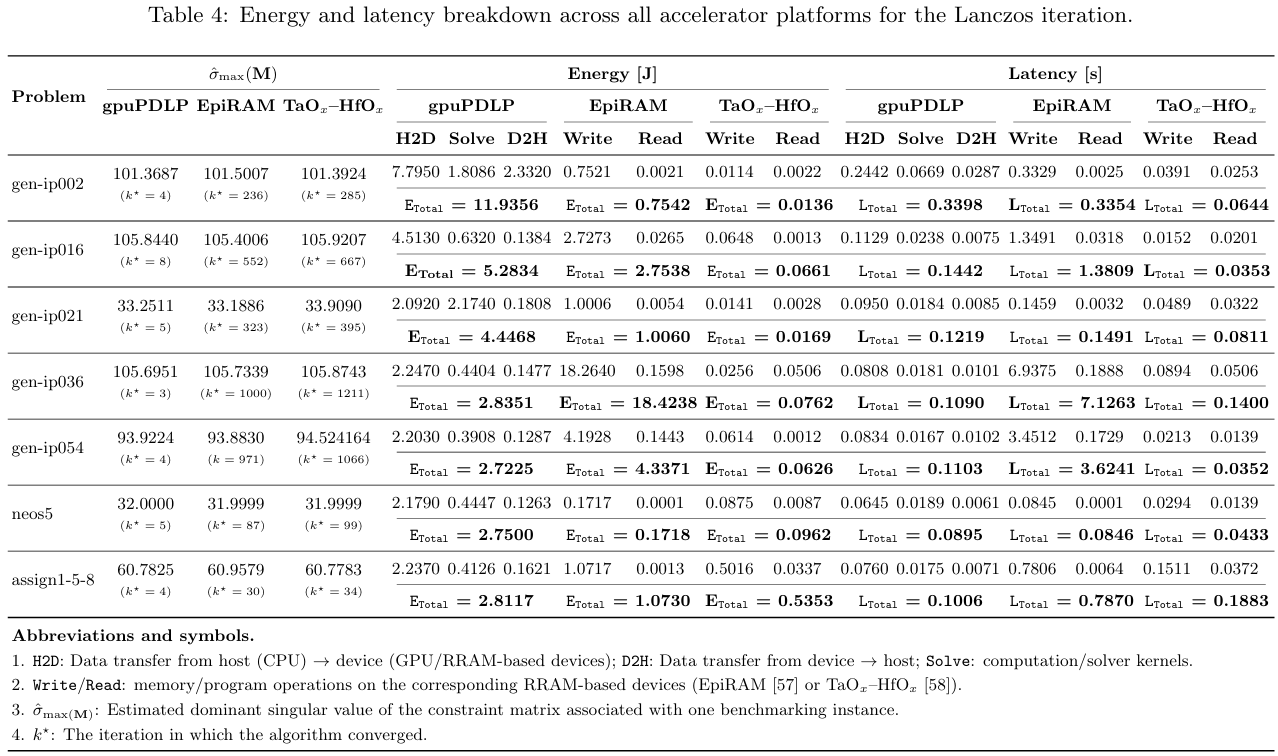}
\end{figure*}
  
\begin{figure*}[!htbp]
\centering
\includegraphics[width=\linewidth]{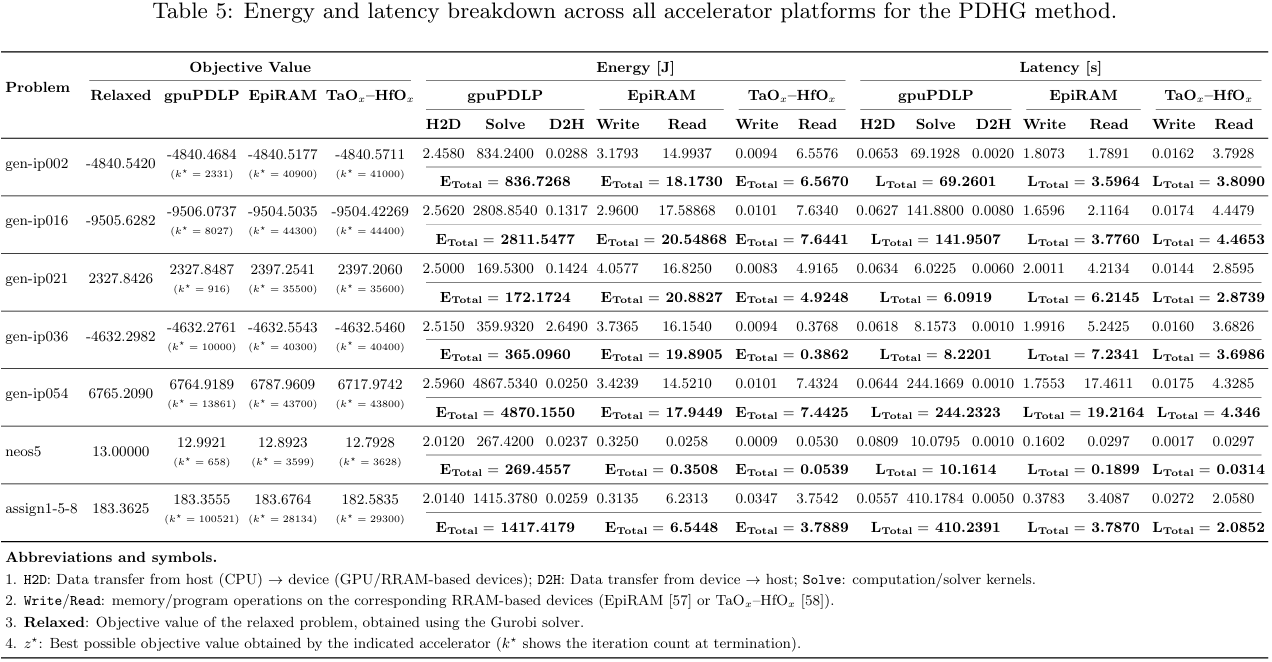}
\end{figure*}

\subsection{Theoretical Proofs}\label{appendix:supplemental_materials}
\subsubsection{Proof of Proposition~1}\label{subsubsec:proof_estimate_norm}
\begin{proof}
Let $\mathbf{K}$ be an arbitrary $m \times n$ matrix whose rank is $r$, i.e., $\mathbf{K} \in \mathbb{R}^{m \times n}$ with $\operatorname{rk}(\mathbf{K}) = \min(m,n) = r$.
Let $\mathbf{K} = \mathbf{U}\pmb{\Sigma}\mathbf{V}^\top$ and $\mathbf{K}^\top = \mathbf{V}\pmb{\Sigma}\mathbf{U}^\top$ be the singular value decomposition (SVD) of $\mathbf{K}$ and $\mathbf{K}^\top$, respectively.
Thus,
\begin{equation*}
    \begin{cases}
    \mathbf{U}\in\mathbb{R}^{m \times m} \\
    \pmb{\Sigma} = \operatorname{diag}(\sigma_1, \cdots, \sigma_{r}|\sigma_{r+1}, \cdots, \sigma_{n-1}, \sigma_{n}) \\
    \mathbf{V}^\top \in \mathbb{R}^{n \times n}
    \end{cases}
\end{equation*}

As a corollary to SVD, for all $i \leq r$:
\begin{equation*}
    \mathbf{K}\mathbf{v}_i = \sigma_i \mathbf{u}_i, \quad \mathbf{K}^\top\mathbf{u}_i = \sigma_i \mathbf{v}_i
\end{equation*}

We then define the following two vectors of length $(m + n)$, denoted by $\mathbf{w}^{(+)}_i$ and $\mathbf{w}^{(-)}_i$:
\begin{equation*}
    \mathbf{w}^{(+)}_i := \begin{bmatrix} \mathbf{u}_i \\ \mathbf{v}_i \end{bmatrix}, \quad \mathbf{w}^{(-)}_i := \begin{bmatrix} \mathbf{u}_i \\ -\mathbf{v}_i \end{bmatrix}, \quad \forall i \leq r
\end{equation*}

Hence,
\begin{equation*}
    \begin{cases}
        \mathbf{M}\mathbf{w}^{(+)}_i = \mathbf{M}\begin{bmatrix} \mathbf{u}_i \\
\mathbf{v}_i \end{bmatrix} = \begin{bmatrix} \mathbf{K}\mathbf{v}_i \\ \mathbf{K}^\top\mathbf{u}_i \end{bmatrix} =  \begin{bmatrix} \sigma_i\mathbf{u}_i \\ \sigma_i\mathbf{v}_i \end{bmatrix} = \sigma_i\mathbf{w}^{(+)}_i \\ \\
        \mathbf{M}\mathbf{w}^{(-)}_i = \mathbf{M}\begin{bmatrix} \mathbf{u}_i \\ -\mathbf{v}_i \end{bmatrix} = \begin{bmatrix} -\mathbf{K}\mathbf{v}_i \\ \mathbf{K}^\top\mathbf{u}_i \end{bmatrix} =  \begin{bmatrix} -\sigma_i\mathbf{u}_i \\ -\sigma_i (-\mathbf{v}_i) \end{bmatrix}= -\sigma_i\mathbf{w}_i^{(-)} \\
    \end{cases}
\end{equation*}

where $\mathbf{M}:= \begin{bmatrix} \mathbf{0}_{m \times m} & \mathbf{K} \\ \mathbf{K}^\top &\mathbf{0}_{n \times n} \end{bmatrix}$.
Thus, $\pm \sigma_i$ are eigenvalues of $\mathbf{M}$. As a result, $\lambda_{\max}(\mathbf{M}) = \sigma_{\max}(\mathbf{K})$.
\end{proof}

\renewcommand\qedsymbol{$\blacksquare$}

\subsubsection{Proof of Lemma \ref{lem:app-lanczos-perturbed}}\label{subsec:proof-app-lanczos-perturbed}
\begin{proof}
Consider the standard Lanczos recurrence with exact arithmetic,
\[
\mathbf{M}\mathbf{Q}_k = \mathbf{Q}_k \mathbf{T}_k + \mathbf{r}_k \mathbf{e}_k^\top,
\]
where $\mathbf{Q}_k = [\mathbf{q}_1, \ldots, \mathbf{q}_k]$ has orthonormal columns and $\mathbf{T}_k$ is the tridiagonal matrix formed from the recurrence coefficients $\{\alpha_j,\beta_j\}$. 

Under hardware noise, the $j$-th matrix--vector multiplication is perturbed as
\[
\mathbf{w}_j = \mathbf{M}\mathbf{q}_j + \pmb{\zeta}_j, 
\qquad \|\pmb{\zeta}_j\| \le \epsilon_{\max}.
\]
Let $\mathbf{L}_k=\mathbf{M}\mathbf{Q}_k$ represent the noise free version of LHS and $\mathbf{\tilde{L}}_k$ represent the same quantity under hardware noise. As a result, $\mathbf{\tilde{L}}_k = \mathbf{L}_k + \mathbf{Z}_k$, where $\mathbf{Z}_k = [\mathbf{\zeta}_1,\ldots \mathbf{\zeta}_k]$.

The tridiagonal coefficients are computed from $\mathbf{w}_j$ and thus yield a perturbed matrix $\mathbf{\tilde{T}}_k = \mathbf{T}_k + \Delta \mathbf{T}_k$. Therefore, the following relationship holds,
\[
\mathbf{\tilde{L}}_k = \mathbf{Q}_k\mathbf{\tilde{T}}_k +\mathbf{r}_k\mathbf{e}_k^\top
\]
The relationship between $\mathbf{L}_k$ and $\mathbf{T}_k$ under noise can be further refined to state,
\[
\mathbf{L}_k=\mathbf{M}\mathbf{Q}_k \;=\; \mathbf{Q}_k\,\mathbf{\tilde{T}}_k \;+\; \mathbf{r}_k \mathbf{e}_k^\top \;+\; \mathbf{E}_k,
\]
where $\mathbf{E}_k = -\mathbf{Z}_k$ collects the accumulated additive errors and any perturbation to the residual $\mathbf{r}_k$ is inductively absorbed into $\mathbf{E}_k$.
Applying the subadditivity of the spectral norm gives
\[
\|\mathbf{E}_k\| 
\;\le\; \sum_{j=1}^k \|\pmb{\zeta}_j\|
\;\le\; k\,\epsilon_{\max}.
\]
Hence, the perturbation in the Lanczos relation grows at most linearly with the number of iterations. 
\end{proof}


\subsubsection{Proof of Theorem \ref{thm:app-lanczos-ritz}}\label{subsec:proof-app-lanczos-ritz}

\begin{proof}
We first analyze the expected bias in the context of the exact Lanczos process. We can express $\mathbb{E}_{\mathbf{q}_1,\pmb{\zeta}}\Big[|\theta_k - L|\Big]$ as  follows:
\[
\mathbb{E}_{\mathbf{q}_1,\pmb{\zeta}}\Big[|\theta_k - L|\Big] = \mathbb{E}_{\mathbf{q}_1,\pmb{\zeta}}\Big[|\lambda_{\max}(\mathbf{\tilde{T}}_k) -\lambda_{\max}(\mathbf{T}_k) + \lambda_{\max}(\mathbf{T}_k) - L|\Big]
\]
As a consequence, we can decompose the expectation term based on the random variables $\mathbf{q}_1$ and $\mathbf{\zeta}$ as follows: 
\begin{equation*}
\begin{aligned}
\mathbb{E}_{\mathbf{q}_1,\pmb{\zeta}}\Big[|\theta_k - L|\Big] &\leq \underbrace{\mathbb{E}_{\mathbf{q}_1}\Big[|\lambda_{\max}(\mathbf{T}_k) - L|\Big]}_{\text{no perturbation (bias)}} \\
&\qquad+ \underbrace{\mathbb{E}_{\pmb{\zeta}}\Big[|\lambda_{\max}(\mathbf{\tilde{T}}_k)
 - \lambda_{\max}(\mathbf{T}_k)|\Big]}_{\text{with hardware perturbations}}
 \end{aligned}
\end{equation*}

\vspace{0.5em}
\noindent\textbf{Unperturbed case}. For the exact tridiagonal matrix $\mathbf{T}_k$ obtained from $\mathbf{M}$ and
a random starting vector $\mathbf{q}_1$, the expected Ritz value bias satisfies
\[
\mathbb{E}_{\mathbf{q}_1}\Big[|\lambda_{\max}(\mathbf{T}_k) - L|\Big]
\ \le\ C\rho^{(k-1)},
\]
where $\rho = \lambda_2 / \lambda_1 < 1$ and $C$ is a constant on the distribution of $\mathbf{q}_1$ (Theorem 3.1 from \cite{kuczynski1992estimating}).

This result captures the expected geometric convergence of the largest Ritz
value under random initialization.

\vspace{0.5em}
\noindent\textbf{Hardware perturbations case.}
By applying Weyl's inequality on noisy matrix-vector products on the computed tridiagonal matrix $\mathbf{\tilde{T}}_k$, we obtain
\[
|\lambda_{\max}(\mathbf{\tilde{T}}_k) - \lambda_{\max}(\mathbf{T}_k)|
\ \le\ \|\mathbf{\tilde{T}}_k - \mathbf{T}_k\|
\]
Further using Lemma \ref{lem:app-lanczos-perturbed} along with norm invariance under orthonormality, we can say that 
\[
\|\mathbf{\tilde{L}}_k - \mathbf{L}_k\| = \|\mathbf{Q_k(\tilde{T}}_k - \mathbf{T}_k)\| = \|\mathbf{\tilde{T}}_k - \mathbf{T}_k\| = \|\mathbf{E}_k\| \leq k\,\epsilon_{\max}
\]

Applying the triangle inequality and taking expectations over both the random start and the noise,
\begin{equation*}
\begin{aligned}
\mathbb{E}_{\mathbf{q}_1,\pmb{\zeta}}\Big[|\theta_k - L|\Big]
&\leq
\mathbb{E}_{\mathbf{q}_1}\Big[|\lambda_{\max}(\mathbf{T}_k) - L|\Big] \\
&\qquad + \mathbb{E}_{\pmb{\zeta}}\Big[|\lambda_{\max}(\mathbf{\tilde{T}}_k)
 - \lambda_{\max}(\mathbf{T}_k)|\Big] \\
&\leq C\rho^{(k-1)} + k\,\epsilon_{\max}.
\end{aligned}
\end{equation*}

\vspace{0.5em}
\noindent\textbf{Averaged (ergodic) estimate.}
Consider the mean absolute deviation of the averaged estimator:
\[
\mathbb{E}_{\mathbf{q}_1,\pmb{\zeta}}|\bar\theta_K - L|
= \mathbb{E}_{\mathbf{q}_1,\pmb{\zeta}}\left|\frac{1}{K}\sum_{k=1}^K (\theta_k - L)\right|
\le \frac{1}{K}\sum_{k=1}^K \mathbb{E}_{\mathbf{q}_1,\pmb{\zeta}}|\theta_k - L|,
\]
where the inequality follows from Jensen’s inequality.
Substituting the single-step bound,
\[
\begin{aligned}
\mathbb{E}_{\mathbf{q}_1,\pmb{\zeta}}|\bar\theta_K - L|
\le &
\frac{1}{K}\sum_{k=1}^K \Big( C\rho^{(k-1)} + k\,\epsilon_{\max} \Big) \\
=&
\underbrace{\frac{C}{K}\sum_{k=1}^K \rho^{(k-1)}}_{(1)}\;
+\;
\underbrace{\frac{\epsilon_{\max}}{K}\sum_{k=1}^K k}_{(2)}.
\end{aligned}
\]

Term~(1) is a finite geometric sum bounded by a constant multiple of $1/K$,
since $\sum_{k=1}^K \rho^{(k-1)} \le 1/(1-\rho)$.  
Term~(2) simplifies to $\mathcal{O}(K\epsilon_{\max})$ for a looser bound after dividing by $K$ since $\sum_{k=1}^K = K(K+1)/2$. Therefore, combining these two bounds yields
\begin{equation*}
\mathbb{E}_{\mathbf{q}_1,\pmb{\zeta}}|\bar\theta_K - L|
\le \mathcal{O}\!\left(\frac{1}{K}\right)
   + \mathcal{O}(K\epsilon_{\max}),
\end{equation*}
which completes the proof.
\end{proof}

\subsubsection{Proof of Lemma \ref{lem:app-stepsize}}\label{subsec:proof-app-stepsize}
\begin{proof}
We start with the assumption about the relative error of our estimate:
    $$
    |\hat{L}-L|\le \bar{\delta} L
    $$
By rewriting as a two-sided inequality and then adding $L$ to all parts, we obtain a bound on where the true $L$ is relative to our estimate $\hat{L}$:
    $$
    \begin{aligned}
        &\qquad \text{-}\bar{\delta} L \le \hat{L}-L \le \bar{\delta} L \\
        &\therefore L - \bar{\delta} L \le \hat{L} \le L +\bar{\delta} L \\
        &\therefore L(1 - \bar{\delta}) \le \hat{L} \leq L(1 + \bar{\delta}) 
    \end{aligned}
    $$
The convergence of the PDHG depends on setting the coupled primal-dual step sizes, $\tau\sigma$, to conform with $1/\hat{L}^2$. We can get a bound for this term from the left side of the inequality as follows:

\begin{enumerate}
    \item Since all terms are positive, we can take the reciprocal of both sides, which flips the inequality sign:
        $$
        \frac{1}{\hat{L}} \le \frac{1}{L(1-\bar\delta)}
        $$
    \item Squaring both sides preserves the inequality:
        $$\frac{1}{\hat{L}^2} \le \frac{1}{L^2(1-\bar\delta)^2}
        $$
    \item By substituting the rule for $\tau\sigma$—that is, $\tau\sigma L^2 = (\frac{\vartheta}{\hat{L}^2})L^2$, and employing the inequality we obtained from the previous step, we now have:
        $$
        \tau\sigma L^2 \leq \Big(\frac{1}{L^2(1-\bar{\delta})^2}\Big)\vartheta L^2 \implies \tau\sigma L^2 \leq \frac{\vartheta}{(1-\bar{\delta})^2}
        $$
\end{enumerate}
For the right side of the inequality, given the safety margin is $\vartheta < (1 - \bar{\delta})^2$, then:
    $$
     \tau\sigma L^2 \leq \frac{\vartheta}{(1-\bar{\delta})^2} \leq 1
    $$
\end{proof}


\subsubsection{Proof of Lemma \ref{lem:one-step-pdhg}}\label{subsec:proof-one-step-pdhg}
\begin{proof}
For the LP saddle-point problem, $f(\mathbf{x}) = \mathbf{c}^\top\mathbf{x}$ and $g(\mathbf{y}) = -\mathbf{b}^\top\mathbf{y}$ over the feasible sets, we have the following results. Using the results of the proximal projection, we have:
    \begin{enumerate}
        \item For any feasible $\mathbf{x}$, the property of the proximal operator gives us:
            $$
            \langle \mathbf{x}^{\{k+1\}} - \Big[\mathbf{x}^{\{k\}} - \tau(\mathbf{c} + \mathbf{K}^\top \mathbf{y}^{\{k\}} + \pmb{\xi}_k)\Big], \mathbf{x} - \mathbf{x}^{\{k+1\}} \rangle \ge 0
            $$
        Rearranging the terms, we get:
            $$
            \begin{aligned}
            &\tau \langle \mathbf{c} + \mathbf{K}^\top \mathbf{y}^{\{k\}} + \pmb{\xi}_k, \mathbf{x}^{\{k+1\}} - \mathbf{x} \rangle \\
            &\quad \le -\langle \mathbf{x}^{\{k\}} - \mathbf{x}^{\{k+1\}}, \mathbf{x} - \mathbf{x}^{\{k+1\}} \rangle
            \end{aligned}
            $$
        Now, we use the identity
\(2\langle \mathbf{a}-\mathbf{b}, \mathbf{b}-\mathbf{c} \rangle
= \|\mathbf{a}-\mathbf{c}\|^2 - \|\mathbf{a}-\mathbf{b}\|^2 - \|\mathbf{b}-\mathbf{c}\|^2\). Let \(\mathbf{a} = \mathbf{x}^{\{k\}}, \mathbf{b} = \mathbf{x}^{\{k+1\}}, \mathbf{c} = \mathbf{x}\),
we obtain:
            $$
            \begin{aligned}
            2\langle \mathbf{x}^{\{k\}} - \mathbf{x}^{\{k+1\}},\mathbf{x}^{\{k+1\}}-\mathbf{x}\rangle = &\|\mathbf{x} - \mathbf{x}^{\{k\}}\|^2 \\
              &- \|\mathbf{x} - \mathbf{x}^{\{k+1\}}\|^2 \\
              & - \|\mathbf{x}^{\{k\}} - \mathbf{x}^{\{k+1\}}\|^2
            \end{aligned}
            $$
        Substituting this into our inequality and dividing by $\tau$, we have:
            $$
            \begin{aligned}
                &\langle \mathbf{c} + \mathbf{K}^\top \mathbf{y}^{\{k\}}, \mathbf{x}^{\{k+1\}} - \mathbf{x} \rangle \le \\ & \qquad \frac{1}{2\tau}\Big(\|\mathbf{x}-\mathbf{x}^{\{k\}}\|^2 -  \|\mathbf{x}-\mathbf{x}^{\{k+1\}}\|^2 - \|\mathbf{x}^{\{k\}} - \mathbf{x}^{\{k+1\}}\|^2\Big)\\ & \qquad -  \langle \pmb{\xi}_k, \mathbf{x}^{\{k+1\}} - \mathbf{x} \rangle
            \end{aligned}
            $$

    \item Similarly, for any feasible dual $\mathbf{y}$ given $\sigma$ the corresponding optimality condition results in the following relation:
            $$
            \begin{aligned} 
            &\langle -\mathbf{b} + \mathbf{K}(2\mathbf{x}^{\{k+1\}}-\mathbf{x}^{\{k\}}), \mathbf{y}^{\{k+1\}} - \mathbf{y} \rangle \le \\
            &\qquad \frac{1}{2\sigma}\Big[\|\mathbf{y}-\mathbf{y}^{\{k\}}\|^2-\|\mathbf{y}-\mathbf{y}^{\{k+1\}}\|^2 - \|\mathbf{y}^{\{k\}} - \mathbf{y}^{\{k+1\}}\|^2 \Big]\\
            &\qquad - \langle \pmb{\zeta}_k, \mathbf{y}^{\{k+1\}} - \mathbf{y} \rangle
            \end{aligned}
            $$
    We specifically note that $\mathbf{b}$ denotes the fixed right-hand side vector in the constraint $\mathbf{Kx} = \mathbf{b}$ of the primal problem.
\end{enumerate}
Adding the two inequalities gives:
\[
\begin{aligned}
&\langle \mathbf{c} - \mathbf{K}^\top \mathbf{y}^{\{k+1\}}, \mathbf{x}^{\{k+1\}} - \mathbf{x}\rangle\\
& \qquad + \langle \mathbf{K}\mathbf{x}^{\{k+1\}} - \mathbf{b}, \mathbf{y}^{\{k+1\}} - \mathbf{y}\rangle \\
&\le
\frac{1}{2\tau}\Big[\|\mathbf{x}-\mathbf{x}^{\{k\}}\|^2 - \|\mathbf{x}-\mathbf{x}^{\{k+1\}}\|^2 - \|\mathbf{x}^{\{k+1\}}-\mathbf{x}^{\{k\}}\|^2\Big]\\
&\quad + \frac{1}{2\sigma}\Big[\|\mathbf{y}-\mathbf{y}^{\{k\}}\|^2 - \|\mathbf{y}-\mathbf{y}^{\{k+1\}}\|^2 - \|\mathbf{y}^{\{k+1\}}-\mathbf{y}^{\{k\}}\|^2\Big]\\
&\quad - \langle \mathbf{K}(\mathbf{x}^{\{k+1\}}-\mathbf{x}^{\{k\}}), \mathbf{y}^{\{k+1\}}-\mathbf{y}\rangle\\
&\quad + \langle \pmb{\xi}_k, \mathbf{x}-\mathbf{x}^{\{k+1\}}\rangle + \langle \pmb{\zeta}_k, \mathbf{y}-\mathbf{y}^{\{k+1\}}\rangle
\end{aligned}
\]
which is exactly the claimed inequality.  

\end{proof}


\subsubsection{Proof of Theorem \ref{thm:ergodic-pdhg}}\label{subsec:proof-ergodic-pdhg}

\begin{proof}
We will proceed by summing the one-step PDHG inequality from Lemma~\ref{lem:one-step-pdhg} over $k=0, \dots, K-1$ and then bounding the resulting terms. We begin with the one-step inequality derived from the properties of the proximal updates. For any feasible solution $(\mathbf{x}^*, \mathbf{y}^*)$, the iterates satisfy:
        $$
        \begin{aligned}
            &\langle \mathbf{c}-\mathbf{K}^\top \mathbf{y}^{\{k+1\}},\mathbf{x}^{\{k+1\}}-\mathbf{x}^*\rangle+\langle \mathbf{K}\mathbf{x}^{\{k+1\}}-\mathbf{b},\mathbf{y}^{\{k+1\}}-\mathbf{y}^*\rangle \\
            &\le \frac{1}{2\tau}\left[\|\mathbf{x}^*-\mathbf{x}^{\{k\}}\|^2-\|\mathbf{x}^*-\mathbf{x}^{\{k+1\}}\|^2 - \|\mathbf{x}^{\{k\}} - \mathbf{x}^{\{k+1\}}\|^2\right] \\
            &+ \frac{1}{2\sigma}\left[\|\mathbf{y}^*-\mathbf{y}^{\{k\}}\|^2-\|\mathbf{y}^*-\mathbf{y}^{\{k+1\}}\|^2 - \|\mathbf{y}^{\{k\}} - \mathbf{y}^{\{k+1\}}\|^2\right] \\
            &- \langle \mathbf{K}(\mathbf{x}^{\{k+1\}} - \mathbf{x}^{\{k\}}),\mathbf{y}^{\{k+1\}}-\mathbf{y}^*\rangle \\
            &+ \langle \pmb{\xi}_k,\mathbf{x}^*-\mathbf{x}^{\{k+1\}}\rangle + \langle \pmb{\zeta}_k,\mathbf{y}^*-\mathbf{y}^{\{k+1\}}\rangle.
        \end{aligned}
        $$
We then proceed towards the treatment of each term on the RHS of the inequality in Lemma \ref{lem:one-step-pdhg}.
\begin{enumerate}
    \item We sum the one-step inequality from $k=0$ to $K-1$. The LHS becomes the sum of the instantaneous gaps. The RHS is a sum of distance-related terms and noise terms, involving squared distances to the optimal solution forming a \textbf{telescoping sum}.
            \begin{itemize}
                \item \textbf{Primal Distances:}
                        \begin{equation*}
                        \begin{aligned}
                            &\sum_{k=0}^{K-1} \frac{1}{2\tau}\left[\|\mathbf{x}^*-\mathbf{x}^{\{k\}}\|^2-\|\mathbf{x}^*-\mathbf{x}^{\{k+1\}}\|^2\right] \\
                            &\qquad = \frac{1}{2\tau}\left[\|\mathbf{x}^*-\mathbf{x}^{\{0\}}\|^2 - \|\mathbf{x}^*-\mathbf{x}^{\{K\}}\|^2\right] \\
                            &\qquad \le \frac{1}{2\tau}\|\mathbf{x}^*-\mathbf{x}^{\{0\}}\|^2
                            \end{aligned} 
                        \end{equation*}
                \item \textbf{Dual Distances:} Similarly for the dual variable:
                        \begin{equation*}
                            \begin{aligned}
                            &\sum_{k=0}^{K-1} \frac{1}{2\sigma}\left[\|\mathbf{y}^*-\mathbf{y}^{\{k\}}\|^2-\|\mathbf{y}^*-\mathbf{y}^{\{k+1\}}\|^2\right] \\
                            &\qquad \le \frac{1}{2\sigma}\|\mathbf{y}^*-\mathbf{y}^{\{0\}}\|^2
                            \end{aligned}
                        \end{equation*}
            \end{itemize}
        The sum of all distance-related terms is therefore bounded by a constant $C_0$ that depends only on the initial point.
        The remaining distance terms including $\Delta\mathbf{x}^{\{k\}} = \mathbf{x}^{\{k+1\}} - \mathbf{x}^{\{k\}}$, $\Delta\mathbf{y}^{\{k\}} = \mathbf{y}^{\{k+1\}} - \mathbf{y}^{\{k\}}$ and the cross terms can be denoted by R, where 
         \begin{equation*}
         \begin{aligned}
           R = &-\frac{1}{2\tau}\sum\limits_{k=0}^{K-1}||\Delta\mathbf{x}^{\{k\}}||^2 - \frac{1}{2\sigma}\sum\limits_{k=0}^{K-1}||\Delta\mathbf{y}^{\{k\}}||^2 \\& \qquad-\sum\limits_{k=0}^{K-1}\langle\mathbf{K}(\Delta\mathbf{x}^{\{k\}}),\mathbf{y}^{\{k+1\}}-\mathbf{y}^*\rangle
        \end{aligned}
        \end{equation*}
        We can upper bound the cross term using the bounded value of the dual $||\mathbf{y}^{\{k+1\}}-\mathbf{y}^*||\leq \mathbf{R}_y$ as follows
                 \begin{equation*}
                 \begin{aligned}
                    &-\sum\limits_{k=0}^{K-1}\langle\mathbf{K}(\Delta\mathbf{x}^{\{k\}}),\mathbf{y}^{\{k+1\}}-\mathbf{y}^*\rangle \\ 
                    & \leq \Big|\sum\limits_{k=0}^{K-1}\langle\mathbf{K}(\Delta\mathbf{x}^{\{k\}}),\mathbf{y}^{\{k+1\}}-\mathbf{y}^*\rangle\Big| \leq \mathbf{L}\mathbf{R}_y\sum\limits_{k=0}^{K-1} ||\Delta\mathbf{x}^{\{k\}}||
            \end{aligned}
        \end{equation*}
        Using Young's inequality $ab\leq \theta a^2/2 + b^2/(2\theta)$ and choosing $\theta=\sigma L$, with $\theta > 0$, we obtain
        \begin{equation*}
                 \begin{aligned}
                    & \mathbf{L}\mathbf{R}_y\sum\limits_{k=0}^{K-1} ||\Delta\mathbf{x}^{\{k\}}|| \leq \tau\sigma L^2\sum\limits_{k=0}^{K-1} \frac{||\Delta\mathbf{x}^{\{k\}}||^2}{2\tau} + \frac{K\mathbf{R}^2_y}{2\sigma}
            \end{aligned}
        \end{equation*}
        Aggregating terms we can see that
                \begin{equation*}
                 \begin{aligned}
                    & R \leq \frac{\tau\sigma L^2-1}{2\tau}\sum\limits_{k=0}^{K-1}||\Delta\mathbf{x}^{\{k\}}||^2  - \frac{1}{2\sigma}\sum\limits_{k=0}^{K-1}||\Delta\mathbf{y}^{\{k\}}||^2 + \frac{K\mathbf{R}^2_y}{2\sigma}
            \end{aligned}
        \end{equation*}
        Since, $\tau\sigma L^2<1$ as established by Lemma \ref{lem:app-stepsize}, the first two terms become non-positive, leaving us with only the last term.
        Therefore, the cross terms add a positive component that creates a total bound including the bounds of distance denoted by $C_K = C_ 0 + \frac{K\mathbf{R}^2_y}{2\sigma}$.
    \item Let $N_K$ be the sum of the noise terms over $K$ iterations:
        \begin{equation*}
        N_K = \sum_{k=0}^{K-1} \Big[ \langle \pmb{\xi}_k,\,\mathbf{x}^*-\mathbf{x}^{\{k+1\}}\rangle + \langle \pmb{\zeta}_k,\,\mathbf{y}^*-\mathbf{y}^{\{k+1\}}\rangle \Big]
        \end{equation*}
        By the law of total expectation and both the unbiasedness and independence assumptions, $\mathbb{E}[N_K] = 0$. To bound the convergence rate, we analyze the expectation of its magnitude, $\mathbb{E}[|N_K|]$. 
        
        Using Jensen's inequality, $\mathbb{E}[|N_K|] \le \sqrt{\mathbb{E}[N_K^2]}$. Since $\mathbb{E}[N_K] = 0$, the standard deviation is $\sqrt{\mathbb{V}(N_K)}$. Due to the independence of noise across iterations, the variance of the sum is the sum of the variances:
            \begin{equation*}
            \begin{aligned}
            \mathbb{V}(N_K) &= \\
            &\sum_{k=0}^{K-1} \mathbb{V}\Big[\langle \pmb{\xi}_k, \mathbf{x}^*-\mathbf{x}^{\{k+1\}}\rangle + \langle \pmb{\zeta}_k, \mathbf{y}^*-\mathbf{y}^{\{k+1\}}\rangle\Big]
            \end{aligned}
            \end{equation*}
        We bound the variance of a single term at iteration $k$. Using Cauchy-Schwarz we have:
            \begin{equation*}
                \begin{aligned}
                        \mathbb{V}(\langle \pmb{\xi}_k, \mathbf{x}^*-\mathbf{x}^{\{k+1\}}\rangle) &= \mathbb{E}[\langle \pmb{\xi}_k, \mathbf{x}^*-\mathbf{x}^{\{k+1\}}\rangle^2] \\
                        &\qquad \le \mathbb{E}[\|\pmb{\xi}_k\|^2 \|\mathbf{x}^*-\mathbf{x}^{\{k+1\}}\|^2]
                \end{aligned}
            \end{equation*}
            With $\mathbb{E}[\|\pmb{\xi}_k\|^2] \le \varphi_\xi^2 \le \delta^2$, and bounded primal $\|\mathbf{x}^*-\mathbf{x}^{\{k+1\}}\|^2 \leq \mathbf{R}^2_x$ this variance is bounded by $(\delta\mathbf{R}_x)^2$. We repeat the same process for $\mathbb{E}[\|\pmb{\zeta}_k\|^2] \le \varphi_\zeta^2 \le \delta^2$ and sum the variances over $K$ iterations, absorbing $\mathbf{R}_x,\mathbf{R}_y$ into the constant $\delta$:
            \begin{equation*}
            \mathbb{V}(N_K) \le \sum_{k=0}^{K-1} 2\delta^2 = 2K\delta^2
            \end{equation*}
            Therefore, the expected magnitude of the noise term is bounded by $\mathbb{E}[|N_K|] \le \sqrt{2 K \delta^2} = \delta\sqrt{2K}$
    
    \item Let $\mathbf{z}^{\{k+1\}} = (\mathbf{x}^{\{k+1\}}, \mathbf{y}^{\{k+1\}})$. We define the restricted gap function evaluated at $\mathbf{z}^*$ as 
        $$
        \begin{aligned}            
        \text{Gap}(\mathbf{z}^{\{k+1\}}) &= \langle \mathbf{c}-\mathbf{K}^\top \mathbf{y}^{\{k+1\}},\mathbf{x}^{\{k+1\}}-\mathbf{x}^*\rangle \\
            &\quad + \langle \mathbf{K}(\mathbf{x}^{\{k+1\}}-\mathbf{b}),\mathbf{y}^{\{k+1\}}-\mathbf{y}^*\rangle,
        \end{aligned}
        $$ 
        which corresponds exactly to the LHS of the inequality in Lemma~\ref{lem:one-step-pdhg}. Since the gap function is convex, applying Jensen's inequality gives us:
        \begin{equation*}
            \begin{aligned}
                &\quad\text{gap}(\mathbf{\bar{z}}_K)=\text{gap}\Big(\frac{1}{K}\sum_{k=0}^{K-1}\mathbf{z}^{\{k\}}\Big) \leq \frac{1}{K}\sum_{k=0}^{K-1}\text{gap}\Big(\mathbf{z}^{\{k\}}\Big) \\
           \end{aligned}
        \end{equation*}
    
    Thus,    
        \begin{equation*}
            K \cdot \mathbb{E}[\operatorname{gap}(\bar{\mathbf{z}}_K)] \le C_0 + \frac{K\mathbf{R}^2_y}{2\sigma} + \delta\sqrt{2K} 
        \end{equation*}
    \item Dividing by $K$ yields the final convergence rate:
        \begin{equation*}
            \mathbb{E}[\operatorname{gap}(\bar{\mathbf{z}}_K)] \le \frac{C_0}{K}  + \frac{\mathbf{R}^2_y}{2\sigma} + \frac{\sqrt{2}\delta}{\sqrt{K}} \equiv \mathcal{O}\left(\frac{1}{K}\right) + \mathcal{O}(1) + \mathcal{O}\left(\frac{1}{\sqrt{K}}\right)
        \end{equation*}
\end{enumerate}
\end{proof}


  

\end{document}